\documentclass[twoside, 11pt]{amsart}
\usepackage{amsmath,amsthm,amssymb,bm,epic}

\usepackage[reset,a4paper,vmargin=3truecm,hmargin=2truecm]{geometry}
\usepackage{amsmath,amsthm,amssymb,amscd,epic}
\usepackage{graphicx,bbm,amsthm,graphicx}
\usepackage{mathtools}
\usepackage{leftidx}
\usepackage{hyperref}
\usepackage{cool}
\usepackage{subcaption}
\usepackage{tikz}
\usepackage[centertableaux]{ytableau}
\usetikzlibrary{decorations.text}

\usepackage{bbold}

\theoremstyle{plain}
\newtheorem{thm}{Theorem}[section]
\newtheorem{lem}[thm]{Lemma}
\newtheorem{prop}[thm]{Proposition}
\newtheorem{cor}[thm]{Corollary}

\theoremstyle{definition}
\newtheorem{dfn}{Definition}[section]
\newtheorem{ex}[thm]{Example}

\theoremstyle{remark}
\newtheorem{rem}{Remark}[section]

\DeclareMathOperator{\tr}{tr}

\DeclareMathOperator{\Spec}{Spec}

\DeclareMathOperator*{\Res}{Res} 

\makeatletter

\newcommand{\HRabi}{H_{\text{\upshape Rabi}}}

\newcommand{\N}{\mathbb{N}} 
\newcommand{\Z}{\mathbb{Z}} 
\newcommand{\Q}{\mathbb{Q}} 
\newcommand{\R}{\mathbb{R}} 
\newcommand{\C}{\mathbb{C}} 

\newcommand{\e}{\varepsilon}

\def\smallunderbrace#1{\mathop{\vtop{\m@th\ialign{##\crcr
   $\hfil\displaystyle{#1}\hfil$\crcr
   \noalign{\kern3\p@\nointerlineskip}%
   \tiny\upbracefill\crcr\noalign{\kern3\p@}}}}\limits}

\DeclareMathOperator*{\regprod}{\mathchoice%
{\ooalign{\hbox{$\displaystyle\prod$}\crcr\hbox{$\displaystyle\coprod$}}}
{\ooalign{\hbox{$\textstyle\prod$}\crcr\hbox{$\textstyle\coprod$}}}
{\ooalign{\hbox{$\scriptstyle\prod$}\crcr\hbox{$\scriptstyle\coprod$}}}
{\ooalign{\hbox{$\scriptscriptstyle\prod$}\crcr\hbox{$\scriptscriptstyle\coprod$}}}}

\newcommand{\bs}{\mathbf{s}}

\newcommand{\btZ}[1]{\Z_{2}^{#1}}

\title[Heat kernel for QRM II]{Heat kernel for the quantum Rabi model II: \\ Propagators and spectral determinants}
\author[C.~Reyes-Bustos and M.~Wakayama]{Cid Reyes-Bustos and Masato Wakayama}

\date{\today}



\begin{document}

\begin{abstract}  

  The quantum Rabi model (QRM) is widely recognized as an important model in quantum systems, particularly in 
  quantum optics. The Hamiltonian $H_{\text{Rabi}}$ is known to have a parity decomposition
  $H_{\text{Rabi}} = H_{+} \oplus H_{-}$. In this paper, we give the explicit formulas for the propagator of the
  Schr\"odinger equation (integral kernel of the time evolution operator) for the Hamiltonian $H_{\text{Rabi}}$
  and $H_{\pm}$ by the Wick rotation (meromorphic continuation) of the corresponding heat kernels.
  In addition, as in the case of the full Hamiltonian of the QRM, we show that for the
  Hamiltonians $H_{\pm}$, the spectral determinant is, up to a non-vanishing entire function, equal to the
  Braak $G$-function (for each parity) used to prove the integrability of the QRM. To do this, we show the
  meromorphic continuation of the spectral zeta function of the Hamiltonians $H_{\pm}$ and give some of its
  basic properties.
\end{abstract}

\maketitle

\tableofcontents

\section{Introduction}
\label{sec:intro}

The quantum Rabi model (QRM) is widely recognized as the simplest and most fundamental model describing quantum light-matter interactions, that is, the interaction between a two-level system and a bosonic field mode. Indeed, it is considered a milestone in the long history of quantum physics \cite{HR2008, KG2004}. In \cite{bcbs2016} we can find a recent collection of introductory, survey and original articles from both experimental and theoretical viewpoints not limited to light-matter interaction but also in diverse fields of research. For a notable achievement in recent experimental studies, we refer the reader to \cite{YS2018}.

The Hamiltonian \(\HRabi\) of the QRM is precisely given by
\[
  \HRabi := \omega a^{\dagger}a + \Delta \sigma_z + g (a + a^{\dagger}) \sigma_x .
\]
Here, \(a^{\dagger}\) and \(a\) are the creation and annihilation operators of the single bosonic mode (\([a,a^{\dagger}]=1 \)), $\sigma_x, \sigma_z$ are the Pauli matrices (sometimes written as \(\sigma_1\) and \(\sigma_3\), but since there is no risk of confusion with the variable \(x\) to appear below in the heat kernel, we use the usual notations), $2\Delta$ is the energy difference between the two levels, $g$ denotes the coupling strength between the two-level system and the bosonic mode with frequency $\omega$ (subsequently, we set $\omega=1$ without loss of generality). The Hamiltonian \(\HRabi\) of the QRM has a $\Z_2$-symmetry that gives the parity decomposition
\[
  H_{\text{Rabi}} = H_{+} \oplus H_{-}.
\]
We note that the integrability of the QRM was established in \cite{B2011PRL} using the $\Z_2$-symmetry. 

The aim of the present paper is to provide a explicit formula for the propagator of the Schr\"odinger equation for the
QRM and the same for the two systems defined by the parity decomposition using the explicit analytical formula of
the corresponding heat kernels \(K_{\text{Rabi}}(x,y,t)\) and \(K_\pm(x,y,t, \Delta)\) obtained
in \cite{RW2020hk}. The explicit formula may be used for precise computation of time evolution of quantum states without the
numerical drawbacks appearing in prior methods (see e.g. \cite{WKB2012}).

We recall that the propagator is the integral kernel associated to the Schr\"odinger equation
corresponding to $\HRabi$. Precisely, it is a (two-by-two matrix valued) function $U_{\text{Rabi}}(x,y,t)$ satisfying
\begin{align*}
& i \frac{\partial}{\partial t} U_{\text{Rabi}}(x,y,t) = -\HRabi U(x,y,t) \quad \text{for all}\quad t>0, \\ 
& \lim_{t \to 0}U_{\text{Rabi}}(x,y,t)=\delta_x(y) {\bf I}_2 \quad \text{for}\quad x,y \in \R.
\end{align*}
In other words, we have the expression as a time evolution operator 
\begin{align*}
U_{\text{Rabi}}(x,y,t)=\langle x\,|\,e^{-itH_{\text{Rabi}}} \,|\,y\rangle (= K_{\text{Rabi}}(x,y,it)) \quad \text{for} \quad t>0,
\end{align*}
whenever the Wick rotation, the analytic continuation of the heat kernel $K_{\text{Rabi}}(x,y,t)$ from real to imaginary time,
exists (possibly with singularities). 

Let us mention works related to the computation of the propagator or heat kernel of QRM.
In \cite{ZZ1988} it has been given an approximated formula for the propagator using path-integral techniques. For the Spin-Boson model, and the QRM as a special case, the Feynman-Kac formula for the heat kernel was obtained in \cite{HH2012,HHL2012} via a Poisson point process and a Euclidean field. However, for the study of longtime behavior of the system, the use of numerical computations or approximations is inevitable (see e.g. \cite{LF2011, CPMP}). Moreover, recently in \cite{DSGKSN2017}, it was shown that even approximation forms which are obtained by a perturbative approach provides can provide a significant insight. In fact, it was shown in \cite{DSGKSN2017} that perturbative diagrammatic approach provides a direct visualization of virtual and physical photons in the physical process using the Jaynes-Cummings (the RWA of the QRM)) propagator.

Our second theme is the study of the spectral (functional) determinant, that is, the zeta regularized product of the spectrum for $H_{\pm}$. The spectral determinant of an operator is a function that has zeros at the eigenvalues, that is, it is the generalization of the characteristic polynomial for a finite matrix, It describes various important topological and/or number theoretical invariants (see e.g. \cite{RS1974, QHS1993TAMS}). In \cite{KRW2017}, a significant relation was found between Braak's $G$-function, the power series of transcendental function used to prove the integrability of the QRM \cite{B2011PRL}, and the spectral zeta function of the QRM \cite{Sugi2016} (the Mellin transform of the partition function). Actually, the $G$-function is (up to a non-vanishing function) equal to the spectral determinant
of \(\HRabi\), that is, the zeta-regularized product associated to the spectral zeta function of the QRM.
This result is significant because, on the one hand, the $G$-function is defined through the solutions of system of ordinary differential equations (which is equivalent with the confluent Heun ODE picture of the QRM) that assures the existence of the entire solutions (see Appendix \ref{sec:bargm-space-confl}) and, on the other hand, the spectral determinant arises from the linear term of the Taylor expansion of the spectral zeta function at the origin. The identification of $G$-function and spectral determinant suggests a deeper relation between the exact solvability of a quantum interaction system and the meromorphic continuation of its spectral zeta function.

We extend this result for the case of the Hamiltonians $H_{\pm}$ of the parity decomposition of the QRM. In order to accomplish this, we show the meromorphic continuation of the spectral zeta function of the Hamiltonians $H_{\pm}$. 
In addition, by using the explicit formula of the partition function $Z_{\text{Rabi}}(t)$ of the QRM
  we obtain a contour integral representation of spectral zeta function $\zeta_{\text{QRM}}(s; \tau)$
  (Theorem \ref{IntRep_SZF}). We leave the proof of the integral representation to Appendix \ref{sec:proofmero}. There, we also discuss certain interesting properties of the spectral zeta function that are the consequence of the contour integral expression.

The common tool behind the two results in this paper, that is, the propagator formula and the identification of the spectral determinant with the $G$-function, is the meromorphic continuation to the complex plane of the heat kernel and partition function allowed by precise estimates obtained from the analytical formulas obtained in \cite{RW2020hk}.

\section{Revisited: the discrete path integral for the heat kernel of QRM} \label{sec:limit}

In this section we recall the formulas for the heat kernel and partition function of
the QRM obtained in \cite{RW2020hk}. 
As an introduction for the reader and to complement the discussion of the aforementioned paper, we give a brief overview of the method of computation for the heat kernel which we refer here as the method of discrete paths. We also provide a new interpretation on the resulting expression of the heat kernel using representation theory.

\subsubsection{Trotter-Kato's product formula}

The Hamiltonian of the QRM is given by
\[
  \HRabi := a^{\dagger}a + \Delta \sigma_z + g (a + a^{\dagger}) \sigma_x,
\]
and it is easy to see that it can be written as
\begin{align*}
  \HRabi = b^{\dagger} b - g^2 + \Delta \sigma_z,
\end{align*}
where $b = a +  g \sigma_x$ and $b^{\dagger} = a^\dag +  g \sigma_x $ are the annihilation and creation operators of
a non-commutative version of the quantum harmonic oscillator (i.e. $[b^{\dagger},b] = \bm{I}_2$).

By the Trotter-Kato product formula (see e.g. \cite{Calin2011}), the heat semigroup is given by
\[
  e^{- t \HRabi} = e^{- t (b^{\dagger}b -g^2 + \Delta \sigma_z)} = \lim_{N\to \infty} (e^{-t (b^{\dagger}b -g^2)/N} e^{-t(\Delta \sigma_z)/N})^N,
\]
with convergence in the strong operator topology, and the heat kernel $K_{\text{Rabi}}(x,y,t)$ is obtained
from this formula. In \S\ref{sec:Propagator}, the propagator $U_{\text{Rabi}}(x,y,t)$ is obtained from the Heat kernel by
the Wick rotation $t \to i t$ from the heat kernel $K_{\text{Rabi}}(x,y,t)$. 

The first step for the computation of the heat kernel is to obtain the explicit form of the integral kernel $K^{(N)}(x,y,t)$ of
$(e^{-t (b^{\dagger}b -g^2)/N} e^{-t(\Delta \sigma_z)/N})^N$. 
Since the kernel $K^{(N)}(x,y,t)$ is a two-by-two matrix-valued function, we can write it
in terms of scalar and a non-commutative (matrix) parts in a reasonable way. More precisely, we write 
\begin{align}
  \label{eq:sumGI}
  K^{(N)}(x,y,t) = \sum_{\bs \in \btZ{N}} G_N(u,\Delta,\bs) I_N(x,y,u, \, \bs),
\end{align}
where $G_N(u,\Delta,\bs)$ is a matrix-valued function and the scalars $I_N(x,y,u, \, \bs)$ correspond to the evaluation of multivariate Gaussian integrals. The matrices appearing in the Hamiltonian $\HRabi$ give rise to the finite group \(\Z_2^N\)  structure of equation \eqref{eq:sumGI}.
An important observation is that the $\Z_2^{N}$-group structure can be interpreted in terms of finite paths as in Figure \ref{fig:paths}, we refer to Section 4.2 of \cite{RW2020hk} for more details. 

By taking the limit, we see that the components of the heat kernel are, up to Mehler's type factor inherited by the
quantum harmonic oscillator, given by
\begin{align}
  \label{eq:limit0}
  & \lim_{N \to \infty}\left( \frac{1-u^{\frac{2\Delta}N}}{2 u^{\frac{\Delta}N}} \right) \sum_{k \geq 3}^N   \left( \frac{1+u^{\frac{2\Delta}N}}{2 u^{\frac{\Delta}N}} \right)^{N-k}  J_\eta^{(k,N)}(x,y,u^{\frac{1}N},g)  \sum_{ \bs \in \Z_2^{k-3}} g^{(\alpha,\eta)}_{k-1}(u^{\frac1N},\bs) R^{(\alpha,\eta)}_\eta(u^{\frac{1}N},\bs),
\end{align}
with $\alpha,\eta \in \{0,1\}$ and where we omitted some trivial terms for clarity (see equation (33) in \cite{RW2020hk}). This form
of the heat kernel suggests a possible evaluation as a type of Riemann integral. However, the presence of multiple changes of
signs depending on $\bs \in \Z_2^{k-3}$ and $k \geq 3$ in the term $g^{(\alpha,\eta)}_{k-1}(u^{\frac1N},\bs) R^{(\alpha,\eta)}_\eta(u^{\frac{1}N},\bs)$ make such an approach unfeasible.

\subsubsection{Fourier transform on $\Z_2^n  \, (n\geq0)$}

At this point the method described in \cite{RW2020hk} greatly differs from the usual evaluation by path integrals.
Let us denote by $|\rho|$ the norm (or length) in $\Z_2^{N}$, that is,
$ |\rho| = \sum \rho_i $ for $\rho=(\rho_1,\rho_2,\cdots,\rho_{N})\in \Z_2^{N}$.
By using Fourier analysis, and more concretely Parseval's identity, on the group algebra \(\C[\Z_2^{k-3}]\), we see that
\[
  \sum_{ \bs \in \Z_2^{k-3}} g^{(\alpha,\eta)}_{k-1}(u^{\frac1N},\bs) R^{(\alpha,\eta)}_\eta(u^{\frac{1}N},\bs) = \sum_{ \rho \in \Z_2^{k-3}}  \hat{g}^{(\alpha,\eta)}_{k-1}(u^{\frac1N},\rho) \hat{R}^{(\alpha,\eta)}_\eta(u^{\frac{1}N},\rho).
\]
Then, by fixing $|\rho|=\lambda \in \Z_{\geq0}$, we observe that in the right-hand side of the above equation the function $\hat{R}^{(\alpha,\eta)}_\eta(u^{\frac{1}N},\rho)$ is given as the exponential of certain $q$-polynomials and that these $q$-polynomials do not depend on $k$. This process resembles the separation of a function defined on $\R^{n}$ into  radial and non-radial parts. We remark, however, that the function $\hat{R}^{(\alpha,\eta)}_\eta(u^{\frac{1}N},\rho)$ is not radial, it is only determined (as a $q$-polynomial) by $|\rho|$. In fact, it turns out that it actually determined by the $\mathfrak{S}_\infty$ orbit of $\rho$ for certain action on $\Z_2^{\infty}$.

\subsubsection{$\mathfrak{S}_\infty$-action on $\Z_2^{\infty}$ (Geometric)}
\label{sec:mathfr-acti-z_2infty}

The infinite sum in the limit appearing in the heat kernel is identified with the inductive limit
\[
  \Z_2^{\infty} =  \varinjlim_{n} \Z_2^{n},
\]
where, for $ i \leq j$, the injective homomorphisms $f_{i j} : \Z_2^i \to \Z_2^j$ are given by
\[
  f_{i j}(\rho) = (\rho_1,\rho_2,\cdots,\rho_i,0,\cdots,0) \in \Z_2^{j}
\]
for $\rho = (\rho_1,\rho_2,\cdots,\rho_i) \in \Z_{2}^{i}$, that is, the natural group embeddings. As usual, we consider the inductive limit $\Z_2^{\infty}$ to
be equipped with the discrete topology.

Let us also consider the infinite symmetric group $\mathfrak{S}_\infty$ obtained by the inductive limit of the finite symmetric groups $\mathfrak{S}_n$ where the inductive homomorphisms are also given by the natural group embeddings. Then $\mathfrak{S}_\infty$ acts naturally on $\Z_2^{\infty}$ and the orbits are given by
\[
  \mathcal{O}_\lambda := \left\{ \sigma \in \Z_2^{\infty} \, : \, |\sigma|= \lambda  \right\}
\]
for $\lambda\geq0$ and where the function $|\cdot|$ is induced by the norms on each group $\Z_2^{n}\, (n\geq0)$. Equivalently, we have $ \mathcal{O}_{\lambda} = \mathfrak{S}_\infty . \sigma$, for $\sigma \in \Z_2^{\infty} $ with $|\sigma|=\lambda$ for $\lambda \in \Z_{\geq 0}$.

In other words, $|\cdot|$ is an orbit invariant for the action and we get an orbit decomposition of $\Z_2^{\infty}$ by
\begin{equation*} 
  \Z_2^{\infty} = \bigsqcup_{n=0}^{\infty} \mathfrak{S}_{\infty} . [\underbrace{1,1,\cdots,1}_{n}] \xleftrightarrow{ \text{labeled by $|\cdot|$}} \Z_{\geq 0},
\end{equation*}
where 
\[
  [\underbrace{1,1,\cdots,1}_{n}] := (\underbrace{1,1,\cdots,1}_{n},0,0,\cdots) \in \Z_2^{\infty}
\]
is the image of $(1,1,\cdots,1) \in \Z_2^{n}$ in $\Z_2^{\infty}$. In Figure \ref{fig:paths} we give an example of an element in $\Z_2^{\infty}$ and its corresponding canonical coset representative. We note that for the orbit $\mathcal{O}_\lambda$, the value $\lambda$ of the invariant coincides with the length of the canonical element $(1,1,\cdots,1)\in \Z_2^\lambda$ described above.

\begin{figure}[!ht]
  \centering
  \begin{tikzpicture}[domain=0:4]


    \draw[step=1,very thin,color=gray] (0.6,2) grid (9,3);

    \draw[thick] (0.6,2) -- (1,2) ;
    \draw[thick] (1,2) -- (9.1,2) 
    node[pos=1,below] {$9$}
    node[pos=0,below]{1} node[pos=0.125,below]{2} node[pos=0.25,below]{3}
    node[pos=0.375,below]{4} node[pos=0.5,below]{5} node[pos=0.625,below]{6}
    node[pos=0.750,below]{7} node[pos=0.875,below]{8}; 

    
    \draw[thick,dashed] (9.1,3) -- (11,3) ;
    \draw[thick,dashed] (9.1,2) -- (11,2) ;

    \draw[thick] (0.6,1.8) node[below] {$0$} -- (0.6,3.2) node[above] {$1$};

    \draw[color=blue,very thick]
    (0.6,3) -- (1,3) to[out=0,in=180] (2,2) to[out=0,in=180] (3,3) to[out=0,in=180]
    (4,3) to[out=0,in=180] (5,3) to[out=0,in=180] (6,2) to[out=0,in=180] (7,3)
    to[out=0,in=180] (8,3) to[out=0,in=180] (9,2) --  (11,2); 


    \node at (1.5,0.5) {$\equiv$};

    
    \draw[step=1,very thin,color=gray] (2.6,0) grid (11,1);

    \draw[thick] (2.6,0) -- (3,0) ;
    \draw[thick] (3,0) -- (11.1,0) 
    node[pos=1,below] {$9$}
    node[pos=0,below]{1} node[pos=0.125,below]{2} node[pos=0.25,below]{3}
    node[pos=0.375,below]{4} node[pos=0.5,below]{5} node[pos=0.625,below]{6}
    node[pos=0.750,below]{7} node[pos=0.875,below]{8}; 

    
    \draw[thick,dashed] (11.1,1) -- (13,1) ;
    \draw[thick,dashed] (11.1,0) -- (13,0) ;

    \draw[thick] (2.6,-0.2) node[below] {$0$} -- (2.6,1.2) node[above] {$1$};

    \draw[color=orange,very thick]
    (2.6,1) -- (3,1) to[out=0,in=180] (4,1) to[out=0,in=180] (5,1) to[out=0,in=180]
    (6,1) to[out=0,in=180] (7,1) to[out=0,in=180] (8,1) to[out=0,in=180] (9,0)
    to[out=0,in=180] (10,0) to[out=0,in=180] (11,0) --  (13,0); 


    \node at (14.2,0.5) {$\pmod{\mathfrak{S}_{\infty}}$};
    
  \end{tikzpicture}
  \caption{A path in $\Z_2^{9} \subset \Z_2^{\infty} $ (above) and the corresponding canonical  $\mathfrak{S}_{\infty}$-orbit representative in $\Z_2^{\infty}$ (below) of the orbit $\mathcal{O}_6$.}
  \label{fig:paths}
\end{figure}
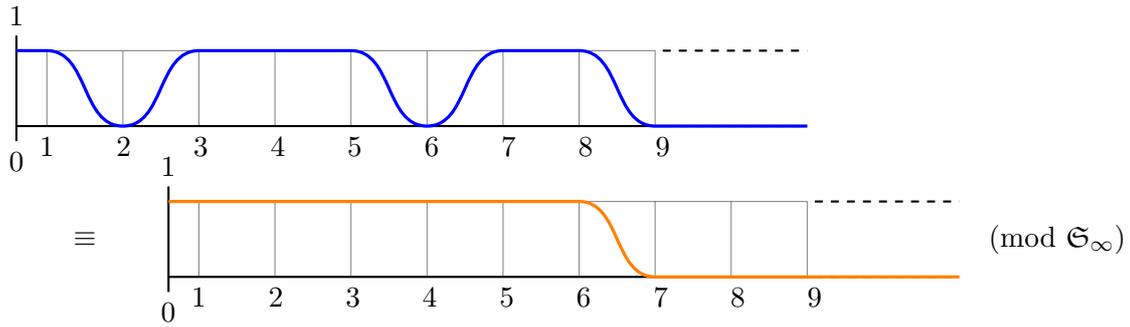

Thus, using this idea, we rearrange the limit for the main body of heat kernel in the following way
\begin{equation}
  \label{eq:limit1}
  \sum_{\lambda=0}^{\infty} \lim_{N \to \infty} \left( \frac{1-u^{\frac{2\Delta}N}}{2 u^{\frac{\Delta}N}} \right) \sum_{k \geq \lambda}^N   \left( \frac{1+u^{\frac{2\Delta}N}}{2 u^{\frac{\Delta}N}} \right)^{N-k}  h_\eta^{(k,N)}(x,y,u^{\frac{1}N},g) \int_{\mathcal{O}_\lambda^{k}}  f^{(\alpha,\eta)}_{k-1}(u^{\frac{1}N},\mu) d\mu_{\lambda},
\end{equation}
for certain functions $h_\eta^{(k,N)}(x,y,u^{\frac{1}N},g)$ and $f^{(\alpha,\eta)}_{k-1}(u^{\frac{1}N},\mu) d\mu_{\lambda}$ and where the innermost integral is the orbital integral of the action and $\mathcal{O}_\lambda^{k} := \mathcal{O}_\lambda \cap \Z_2^k$ (here, $\Z_2^k$ is regarded as a subgroup of $\Z_2^{\infty}$).
Let us describe the orbital integral. First, notice that when we fix $|\rho| = \lambda$, the elements of $\mathcal{O}_\lambda^{k}$ are determined by the position of the ones, in other words there is a bijection
\[
  \left\{ \rho \in \Z_2^k \, : \, |\rho|= \lambda \right\} \longleftrightarrow \left\{ (j_1,j_2,\cdots,j_\lambda) \in \Z_{\geq1}^{\lambda} \,;\, j_1 <j_2<\cdots<j_\lambda \leq k \right\}.
\]
For instance, if we write $\tfrac{j}N$ to indicate the presence of a $1$ in the $j$-th position of a given $\rho \in \Z_2^N$, the elements of $\mathcal{O}_\lambda^{N}$ are given by
\[
  \left(\tfrac{1}N, \tfrac{2}N, \cdots, \tfrac{\lambda}{N}\right), \left(\tfrac{1}N, \tfrac{3}N,\tfrac{4}{N} \cdots, \tfrac{\lambda+1}{N}\right), \left(\tfrac{1}N, \tfrac{2}N,\tfrac{4}{N} \cdots, \tfrac{\lambda+1}{N}\right),\cdots,\left(\tfrac{N-\lambda}N, \tfrac{N-\lambda+1}N,\cdots, \tfrac{N}{N}\right) 
\]
for $\lambda << N$. The orbit integral is then realized by a sum
\[
  \int_{\mathcal{O}_\lambda^{k}}  f^{(\alpha,\eta)}_{k-1}(u^{\frac{1}N},\mu) d\mu_{\lambda} = \sum_{1\le j_1 < j_2 < \cdots < j_\lambda\leq k} f^{(\alpha,\eta)}_{k-1}(u^{\frac{1}N},\bm{j}), 
\]
and moreover, the sum in the right hand side is identified, up to certain terms of bounded order,  by an integral
on the $\lambda$-simplex by the usual procedure using the Riemann-Stieljes integration
\begin{equation}
  \label{eq:simplexint}
  \sum_{1\le j_1 < j_2 < \cdots < j_\lambda\leq k} f^{(\alpha,\eta)}_{k-1}(u^{\frac{1}N},\bm{j}) = \int_{0}^k \int_0^{z_\lambda} \cdots \int_0^{z_{2}} f^{(\alpha,\eta)}_{k-1}(u^{\frac{1}N},\bm{z})  d \bm{z} + O(k^{\lambda-1}).
\end{equation}
 Finally, by using an appropriate change of variables, we are able to evaluate the limit in equation
\eqref{eq:limit1} as a Riemann sum, thus obtaining the explicit formula for the heat kernel.
The lower order terms in \eqref{eq:simplexint} vanish at the
limit appearing in the heat kernel.  

To summarize, in place of Feynman or Feynman-Kac path integrals, in this method we consider only the ``discrete paths'' in the inductive limit  $\Z_2^{\infty}$. In fact, the heat kernel $K_{\text{Rabi}}(x,y,t)$ can be expressed as the summation over the orbit of the infinite
symmetric group $\mathfrak{S}_{\infty}$ on $\Z_2^{\infty}$ by the decomposition described above.
We would like to emphasize the role of the harmonic analysis on the inductive family of groups $\{\Z_2^n\}_{n\geq0}$. Concretely, it allows us to transform the (infinitely many) changes of sign in the limit \eqref{eq:limit0} into an expression that is evaluated into a series of the form
\[
  \sum_{\lambda\geq 0} (t \Delta)^{2 \lambda} \int_{ ( \mathfrak{S}_{\infty}.[1,1,\cdots,1] )/ \mathfrak{S}_\lambda  } \exp\left(\text{hyperbolic functions of $t, \bm{\mu}_\lambda$} \right) d \bm{\mu}_{\lambda},
\]
where the orbital integral over $(\mathfrak{S}_{\infty} . [1,1,\cdots,1]) / \mathfrak{S}_\lambda $ 
is realized as an integral over the $\lambda$-th simplex.

\subsubsection{$\mathfrak{S}_\infty$-module decomposition for $\Z_2^{\infty}$ (Algebraic)}

In the foregoing discussion, we  described geometrically the discrete path integral based
on the group $\Z_2^{\infty}$. However, since the computation is effectively done in the dual space, that is,
the space of its Fourier image, which is also identified by $\Z_2^{\infty}$, it would be natural to understand
this discrete path integral in the framework of the representation theory of the infinite symmetric group
$\mathfrak{S}_\infty$ (see e.g. \cite{Ol2001, TV2007}).

Actually, the representative $[1,1,\ldots,1]$ in $\Z_2^{\infty}$ can be regarded as the partition $\Pi_{\lambda}=(\lambda, \infty)$ of $\N$.
Then the corresponding Young subgroup $\mathfrak{S}_{\Pi_{\lambda}}$ of $\mathfrak{S}_\infty$ is given by
$\mathfrak{S}_{\Pi_{\lambda}}= \tilde{\mathfrak{S}}_\lambda \times \tilde{\mathfrak{S}}_\infty $, where $\tilde{\mathfrak{S}}_\lambda$ is isomorphic to
$\mathfrak{S}_\lambda$ in $\mathfrak{S}_\infty$, that is, its tail is infinite identity permutation. Also $\tilde{\mathfrak{S}}_\infty $
is isomorphic to $\mathfrak{S}_\infty $ by the first $\lambda$ component is the trivial permutation. Notice that the pair of Young diagrams
\[
  \bigg(\, {\underbrace{\ydiagram{2}\cdots\ydiagram{2}}_{\lambda} \; , \; \underbrace{\ydiagram{5}\cdots}_{\infty}}\bigg)
\]
correspond to the trivial representation of $\mathfrak{S}_{\Pi_{\lambda}}$.
Then the representation $\mathrm{Ind}_{\mathfrak{S}_{\Pi_{\lambda}} }^{\mathfrak{S}_\infty }1$ induced from the identity representation
of $\mathfrak{S}_{\Pi_{\lambda}}$ is irreducible (the most simple example of the type I representation in \cite{TV2007}) and as
a $\mathfrak{S}_\infty$-module
\[
  \mathfrak{S}_\infty \, {}^\curvearrowright \, \Z_2^{\infty}=\oplus_{\lambda=0}^\infty \mathrm{Ind}_{\mathfrak{S}_{\Pi_{\lambda}} }^{\mathfrak{S}_\infty }1 .  
\]
This decomposition gives another interpretation of the expression for the heat kernel. Therefore, one of the ways to explore the relation between the discrete path integral and the Feynman path integral is to study the asymptotic combinatorics as developed in e.g. \cite{Hora2001}.

\begin{rem} 
  In the setting of discrete paths, we expect the existence of an equivalence relation $\sim$ of (continuous) paths such that
\begin{equation}
  \label{eq:equiv}
  \{\text{paths}\}/\sim \, = \, \Z_2^{\infty},
\end{equation}
in a way that the path integral can be identified with the ``discrete path integral''. The existence of the equivalence relation $\sim$ is highly non-trivial since, as it is well-known, the set of paths $\{\text{paths}\}$ is uncountable and, in general, there is no known  way to associate a reasonable measure to the path integral, while the set of discrete (countably many) paths $\Z_2^{\infty}$ has a well behaved topology induced from the point measure in each group $\Z_2^{n}$. Thus, further understanding of the conjectural relation $\sim$ may be of importance for the mathematical formalization of the path integral for a certain class of quantum systems.  In addition, the derivation of the discrete path integral, that is, the explicit formula of the heat kernel, from the concrete Feynman-Kac formula for the QRM (obtained for the QRM in \cite{HH2012}) may also be an interesting problem. In Figure \ref{fig:paths2} we give a schematic picture of this discussion. 

\begin{figure}[!ht]
  \centering
  \begin{tikzpicture}[domain=0:4]


    \node at (0,2.5) { {\small\begin{tabular}{c}  Trotter-Kato \\ product formula \end{tabular}}};

    \node at (7,5) { {\small Path integral}};
    \node at (7,3.9) {{\Large $\bigcup$}};
    
    \node at (7,2.6) { {\small\begin{tabular}{c}  Feynman-Kac \\ path integral \end{tabular}}};

    \node at (7,0) { {\small\begin{tabular}{c}  Discrete \\ path integral \end{tabular}}};

    \node at (5.5,1.1) {\begin{tabular}{c} ${\scriptstyle \Z_2^{\infty}}$ {\tiny harmonic}  \\ {\tiny analysis } \end{tabular}};
    

    \draw[->,dashed,thick] (1.8,2.5) --  (5.4,5); 
g    

    \draw[->,thick] (1.8,2.5) -- (5.4,0);

    \draw[->, dashed,thick,double] (7,1.9) --  (7,0.7); 

    \draw[->] (7.2,1.5) -- (11,0.9);


    \draw[->,dashed,thick] (8.3,5) to[out=0,in=0] (8.5,0);
    
    \node at (11.5,2.5) { {\small$ \{\text{paths}\}/\sim = \Z_2^{\infty}  \; \bm{?}$}};


    \node[fill=green!20,draw] at (1.4,0.5) { {\tiny \begin{tabular}{c}  Gaussian integrations \\ + \\ Fourier analysis on $\Z_2^{n} (n\geq 0)$ \\ (dual space picture) \end{tabular}}};
    
    \node[fill=green!20,draw] at (12.1,0.5) { {\tiny \begin{tabular}{c}  $\mathfrak{S}_{\infty}$-orbital integrals \\ with probability measure \\ (asymptotic combinatorics) \end{tabular}}};
    
  \end{tikzpicture}
  \caption{Schematic picture of computation of the heat kernel for the QRM}
  \label{fig:paths2}
\end{figure}
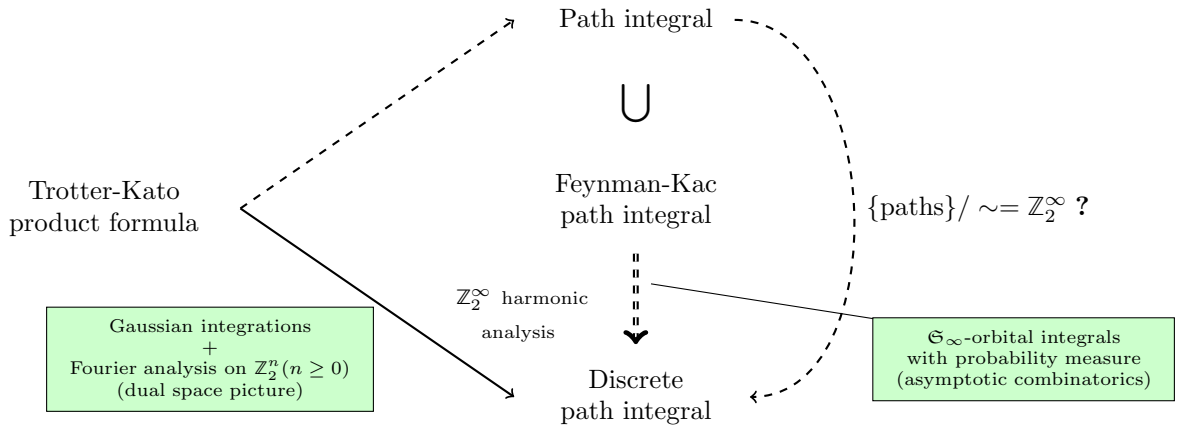

We remark that since the $\Z_2^{\infty}$-group structure is inherent to the Hamiltonian $\HRabi$, the conjectural equivalence \eqref{eq:equiv} may be determined only up to the given quantum system. It is difficult to expect a unique relation of the type \eqref{eq:equiv} for path integrals in general.
\end{rem}

\begin{rem}
  The $\Z_2^{\infty}$ group structure appearing in the computation of the heat kernel is unrelated to 
  the $\Z_2$-parity of the QRM. In fact, it is not difficult to note that a similar $\Z_2^{\infty}$ structure also
  appears in the case of the asymmetric quantum Rabi model
  \[
    \HRabi^{\e} := a^{\dagger}a + \Delta \sigma_z + g (a + a^{\dagger}) \sigma_x + \e \sigma_x,
  \]
  even thought a $\Z_2$-parity is known only for the case $\e = 0$ (QRM case).
\end{rem}

\subsection{Explicit formulas for heat kernel and partition function}
\label{sec:expl-form-heat}


\subsubsection{Heat kernel}
\label{sec:heat-kernel}

In the expressions for the heat kernel and the partition function, the
integral over the $\lambda$-th simplex for \(\lambda = 0 \) is used with the meaning
\[
  \idotsint\limits_{0\leq \mu_1 \leq \cdots \leq \mu_\lambda \leq 1} f(x) d \bm{\mu_0} = f(x),
\]
for any function \(f(x)\).

\begin{thm}[Thm 4.2 of \cite{RW2020hk}] 
  The heat kernel $K_{\text{Rabi}}(x,y,t)$ of the QRM is given by the uniformly convergent series
  \begin{align*}    
    &K_{\text{Rabi}}  (x,y,t) =  K_0(x,y,g,t) \Bigg[ \sum_{\lambda=0}^{\infty} (t\Delta)^{\lambda} e^{-2g^2 (\coth(\tfrac{t}2))^{(-1)^\lambda}}
    \\
    &\qquad\times \idotsint\limits_{0\leq \mu_1 \leq \cdots \leq \mu_\lambda \leq 1} e^{4g^2 \frac{\cosh(t(1-\mu_\lambda))}{\sinh(t)}(\frac{1+(-1)^\lambda}{2}) + \xi_{\lambda}(\bm{\mu_{\lambda}},t)}  
          \begin{bmatrix}
            (-1)^{\lambda} \cosh  &  (-1)^{\lambda+1} \sinh  \\
            -\sinh &  \cosh
          \end{bmatrix}
                     \left( \theta_{\lambda}(x,y,\bm{\mu_{\lambda}},t) \right) d \bm{\mu_{\lambda}} \Bigg],
  \end{align*}
  with \(\bm{\mu_0} := 0\) and \(\bm{\mu_{\lambda}}= (\mu_1,\mu_2,\cdots,\mu_\lambda)\) and \(d \bm{\mu_{\lambda}} = d \mu_1 d \mu_2 \cdots d \mu_{\lambda} \)
  for \(\lambda \geq 1\).

  Here, $K_0(x,y,g,t)$ is given by
  \begin{align*}
    K_0(x,y,g,t)
    & := \frac{e^{t(g^2+\tfrac12)}}{\sqrt{2\pi \sinh(t)}} \exp\left( -\frac{(x^2 + y^2) \cosh(t) - 2x y}{2\sinh(t)} \right)\\
  \end{align*}
  and the functions \(\theta_{\lambda}(x,y, \bm{\mu_{\lambda}},t)\) and $\xi_\lambda(\bm{\mu_{\lambda}},t)$ are given by
  \begin{align*} 
    \theta_{\lambda}(x,y, \bm{\mu_{\lambda}},t) &:= \frac{2\sqrt{2} g}{\sinh(t)}\left( x \cosh(t) - y \right) \left( \frac{1-(-1)^{\lambda}}{2} \right) - \sqrt{2} g (x-y) \coth(\tfrac{t}2) \\
    & \quad +   \frac{2\sqrt{2} g (-1)^{\lambda} }{\sinh(t)}  \sum_{\gamma=0}^{\lambda} (-1)^{\gamma} \Big[ x  \cosh(t(1 -   \mu_{\gamma})) -  y  \cosh(t \mu_{\gamma})  \Big] \nonumber \\
    \xi_\lambda(\bm{\mu_{\lambda}},t) &:=  -\frac{8g^2 }{\sinh(t)} \left(\sinh(\tfrac12t(1-\mu_\lambda))\right)^2 (-1)^{\lambda}  \sum_{\gamma=0}^{\lambda} (-1)^{\gamma} \cosh( t \mu_{\gamma}) \nonumber  \\
   &\qquad  - \frac{4 g^2  }{\sinh(t)} \sum_{\substack{0\leq\alpha<\beta\leq \lambda-1\\ \beta - \alpha \equiv 1 \pmod{2}  }}  \left( \cosh(t(\mu_{\beta+1}-1)-\cosh(t(\mu_{\beta}-1)) \right) \nonumber \\
   &\qquad \qquad \qquad \qquad \qquad \qquad \times ( \cosh(t  \mu_{\alpha}) - \cosh(t \mu_{\alpha+1})), \nonumber 
  \end{align*}
where we use the convention \( \mu_0 = 0 \) whenever it appears in the formulas above.
\end{thm}

As mentioned in the Introduction, the Hamiltonian of the QRM has a parity decomposition \(\HRabi = H_{+}\oplus H_{-} \).
The formula for the heat kernel $K_{\pm}$ of $H_{\pm}$ can be obtained directly from the analytical formula of heat kernel
of the QRM, this is the method used in \cite{RW2020hk}. A direct computation of the heat kernel for the Hamiltonians $H_{\pm}$ is also possible by using the method described in \S\ref{sec:limit} but appears to be more complicated. A brief overview of the decomposition in the Bargmann space and $G$-functions for the QRM can be found in Appendix  \ref{sec:bargm-space-confl}.

\begin{thm}[Thm 4.4 of \cite{RW2020hk}] 
  The heat kernel $K_{\pm}(x,y,t, \Delta)$ of $H_\pm= \HRabi |_{\mathcal{H}_\pm}$ is given by 
  \begin{align*}
    K_{\pm}(x,y,t, \Delta)
    = & K_0(x,y,g,t)\sum_{\lambda=0}^{\infty} (t\Delta)^{2\lambda} \Phi^-_{2\lambda}(x,y,t) \mp \tilde{K}_0(x,-y,g,t) \sum_{\lambda=0}^{\infty}
        (t\Delta)^{2\lambda+1} \Phi^+_{2\lambda+1}(x,-y,t),
  \end{align*}
  where for $\lambda\geq1$, the function $\Phi^\pm_{\lambda}(x,y,t)$ is given by
  \begin{align*}
    \Phi^\pm_{\lambda}(x,y,t) := e^{-2g^2 (\coth(\tfrac{t}2))^{(-1)^\lambda}} \idotsint\limits_{0\leq \mu_1 \leq \cdots \leq \mu_{\lambda} \leq 1}  e^{4g^2 \frac{\cosh(t(1-\mu_\lambda))}{\sinh(t)}(\frac{1+(-1)^\lambda}{2}) +  \xi_{\lambda}(\bm{\mu_{\lambda}},t)\pm \theta_{n}(x,y, \bm{\mu_{\lambda}},t)} d \bm{\mu_{n}}
  \end{align*}
  and  
  \begin{align*}
    \Phi^\pm_0(x,y,t) := e^{-2g^2\tanh\big(\frac{t}2\big) \pm\sqrt2 g(x+y)\tanh\big(\frac{t}2\big)}.
  \end{align*}
\end{thm}

It is not difficult to verify that in fact, we have
\[
  K_{\text{Rabi}}(x,y,t, \Delta)= K_{+}(x,y,t, \Delta)\oplus  K_{-}(x,y,t, \Delta).
\]

\subsubsection{Partition function}
\label{sec:partition-function}

The partition function \( Z_{\text{Rabi}}(\beta)\) for the QRM is obtained by direct computation from the formula of the heat kernel $K_{\text{Rabi}}(x,y,t)$ by the identity
\begin{equation*}
  Z_{\text{Rabi}}(\beta):=  \int_{-\infty}^\infty  \tr K_{\text{Rabi}}(x,x,\beta) dx.
\end{equation*} 
The partition functions \(Z_{\rm{Rabi}}^{\pm}(\beta) \) of the Hamiltonians $H_{\pm}$ are obtained in an analogous way.

\begin{cor}[Cor. 4.3 of \cite{RW2020hk}] \label{cor:Partition_function}
  The partition function \( Z_{\text{Rabi}}(\beta)\) of the QRM is given by
    \begin{align*}
    Z_{\text{Rabi}}(\beta) &= \frac{e^{\beta(g^2+1)}}{\sinh(\beta)} \Bigg[ 1 + e^{-2g^2 \coth(\frac{\beta}2)} \sum_{\lambda=1}^{\infty} (\beta \Delta)^{2\lambda} \idotsint\limits_{0\leq \mu_1 \leq \cdots \leq \mu_{2 \lambda} \leq 1} e^{ 4g^2\frac{\cosh(\beta(1-\mu_{2\lambda}))}{\sinh(\beta)} +  \xi_{2 \lambda}(\bm{\mu_{2\lambda}},\beta) +\psi^-_{2 \lambda}(\bm{\mu_{2 \lambda}},\beta)} d \bm{\mu_{2 \lambda}}  \Bigg],
  \end{align*}
  where the function $\psi_\lambda^{-}(\bm{\mu_{\lambda}},t)$ is given by
  \begin{equation*}
  \psi_\lambda^{-}(\bm{\mu_{\lambda}},t) :=  \frac{4 g^2 }{\sinh(t)}\left[ \sum_{\gamma=0}^{\lambda} (-1)^{\gamma}  \sinh(t\left(\tfrac12 - \mu_{\gamma}\right)  \right]^2.
\end{equation*}
for  \(\lambda \geq 1\) and \(\bm{\mu_{\lambda}} = (\mu_1,\mu_2,\cdots,\mu_\lambda) \) and where \( \mu_0 = 0 \). 
\end{cor}

\begin{cor}[Cor. 4.4 of \cite{RW2020hk}] 
  The partition function \(Z_{\rm{Rabi}}^{\pm}(\beta) \) for the Hamiltonian \(H_{\pm}\) is given by
  \begin{align*}
    &Z_{\text{Rabi}}^{\pm}(\beta) = \frac{ e^{\beta(g^2+1)}}{2\sinh(\beta)} \Bigg[ 1 + e^{-2g^2 \coth(\frac{\beta}2)}  \sum_{\lambda =1}^{\infty} (\beta\Delta)^{2 \lambda} \idotsint\limits_{0\leq \mu_1 \leq \cdots \leq \mu_{2 \lambda} \leq 1} e^{4g^2\frac{\cosh(\beta(1-\mu_{2\lambda}))}{\sinh(t)} +   \xi_{2\lambda}(\bm{\mu_{2 \lambda}},\beta) +\psi^-_{2 \lambda} (\bm{\mu_{2\lambda}},\beta)} d \bm{\mu_{2\lambda}}  \Bigg] \\
     &\qquad \qquad \mp \frac{ e^{\beta(g^2 +1)}}{2\cosh(\beta)} e^{ - 2g^2 \tanh(\frac{\beta}2)}  \sum_{\lambda = 0}^{\infty} (\beta \Delta)^{2\lambda+1} \idotsint\limits_{0\leq \mu_1 \leq \cdots \leq \mu_{2 \lambda+1} \leq 1} e^{\xi_{2\lambda+1}(\bm{\mu_{2 \lambda+1}},\beta) +\psi^+_{2 \lambda+1} (\bm{\mu_{2\lambda+1}},\beta)} d \bm{\mu_{2\lambda+1}},
  \end{align*}
  where the function $\psi_\lambda^{-}(\bm{\mu_{\lambda}},t)$ is as in Corollary \ref{cor:Partition_function} and
 \begin{equation*}
  \psi_\lambda^{+}(\bm{\mu_{\lambda}},t) :=  \frac{4 g^2 }{\sinh(t)}\left[ \sum_{\gamma=0}^{\lambda} (-1)^{\gamma}  \cosh(t\left(\tfrac12 - \mu_{\gamma}\right)  \right]^2.
\end{equation*}
\end{cor}

\section{Propagator of the QRM}
\label{sec:Propagator}

The propagator is the integral kernel associated to the solution of the Schr\"odinger equation corresponding to $\HRabi$.
Precisely, it is a (two-by-two matrix valued) function $U_{\text{Rabi}}(x,y,t)$ satisfying
$i \frac{\partial}{\partial t} U_{\text{Rabi}}(x,y,t)= -\HRabi U_{\text{Rabi}}(x,y,t)$ for all $t>0$ and $\lim_{t \to 0}U_{\text{Rabi}}(x,y,t)=\delta_x(y) \bf{I}_2$ for $x,y \in \R$.

Clearly, we may obtain $U_{\text{Rabi}}(x,y,t)$ from an explicit expression for the propagator from the heat kernel
$K_{\text{Rabi}}$ by a change of variable $t \to i t$ (known as  the Wick rotation).
In this section we formalize this idea by extending the domain $K_{\text{Rabi}}$ to the complex plane as
a function with respect to the variable $t$.

First, a simple technical lemma is needed to establish the holomorphicity of \(K_{\text{Rabi}}(x,y,t)\).
The lemma is also used later to give the meromorphic continuation of the spectral zeta function
of the QRM and the Hamiltonian of each parity.

\begin{lem}
  \label{lem:bound}
  Suppose \(\lambda \in \Z_{\geq 1}\) and let
  \[
    \mathcal{R}^* = \{ z \in \C \, | \, z \neq n \pi i \, , n \in \Z\}
  \]
  Then, for \( t \in \mathcal{R}^* \) there are real valued functions \( C_1(x,y,t), C_2(t), C_3(t) \geq 0 \)
  bounded in compact subsets of \(\mathcal{R}^*\), such that
  \begin{align*}
    \left| \theta_{\lambda}(\bm{\mu_\lambda},x,y,t) \right| &\leq \left|\frac{\sqrt{2} g }{1-e^{-2 t}} \right| C_1(x,y,t) \\
    \left| \psi_{\lambda}^{\pm}(\bm{\mu_\lambda},t) \right| &\leq \left|\frac{2 g^2 }{1-e^{-2 t}} \right| C_2(t) \\
    \left\vert \xi_{\lambda}(\bm{\mu_{\lambda}},t) \right\vert  &\le \left|\frac{2 g^2 }{1-e^{-2 t}} \right| C_3(t) \lambda 
  \end{align*}
  uniformly for \(0 \leq  \mu_1 \leq \mu_2 \leq \cdots \leq \mu_\lambda \leq 1\).
\end{lem}

\begin{proof}
  Let \( t = a+ b i \in \mathcal{R}^*\). For the proof it
    is convenient to use the exponential form of the hyperbolic functions appearing in the functions $\theta_{\lambda}(\bm{\mu_{\lambda}},x,y,t)$,$\psi_{\lambda}^{\pm}(\bm{\mu_{\lambda}},t)$  and $\xi_{\lambda}(\bm{\mu_{\lambda}},t)$.
  Lets consider first the case of $\theta_{\lambda}(\bm{\mu_\lambda},x,y,t)$.  Clearly, we have
  \begin{align*}
    &\left|\frac{2\sqrt{2} g e^{-t}}{1-e^{-2t}}\left( x (e^{t}+e^{- t}) - 2 y \right) \left( \frac{1-(-1)^{\lambda}}{2} \right) - \sqrt{2} g (x-y) \frac{1+e^{-t}}{1-e^{-t}}\right| \\
    & \qquad \qquad \leq \frac{\sqrt{2}|g| e^{-a}}{|1-e^{-2t}|} \left(2 (|x|(e^a + e^{-a}) + 2|y|) + (|x|+|y|)(1+e^{-a})^2 \right) = \frac{\sqrt{2} g c_1(x,y,t)}{|1-e^{-2t}|},
  \end{align*}
  Next, we notice that, for \( \lambda \equiv 1 \pmod{2} \), we have
  \begin{align*}
    &s(x,t,y) := \sum_{\gamma=0}^{\lambda} (-1)^{\gamma} \Big[ x  (e^{t(1 -   \mu_{\gamma}) } + e^{ t( \mu_{\gamma} - 1)})  -  y  (e^{- t \mu_{\gamma} }+ e^{ t \mu_{\gamma}})  \Big] \\
    &\qquad = - t \sum_{\gamma = 0}^{\frac{2\lambda-1}{2}} \left( x \int_{\mu_{2\gamma}}^{\mu_{2\gamma+1}} (e^{t(1-x)} + e^{t(x-1)})d x - y \int_{\mu_{2\gamma}}^{\mu_{2\gamma+1}} (e^{t x} + e^{-t x}) d x  \right).
  \end{align*}
  Next, we have
  \[
    |s(x,t,y)| \leq  |t| \left(|x| \int_{0}^{1} (e^{a(1-x)} + e^{a(x-1)})d x + |y| \int_{0}^{1} (e^{a x)} + e^{-a x)})d x  \right),
  \]
  giving
  \[
    |s(x,y,t)| \leq 2\frac{|t|}{|a|}(|x|+|y|)\left(e^a + e^{-a}\right).
  \]
  If  \( \lambda \equiv 0 \pmod{2} \), we apply the estimate above to the first \(\lambda\) terms of the sum resulting in
  \[
    |s(x,y,t)| \leq 2\frac{|t|}{|a|}(|x|+|y|)\left(2e^{|a|} + e^{-|a|} + 1 \right) 
  \]
  and we set \(c_2(x,y,t)\) as the right hand side of the inequality. Setting $C_1(x,y,t) = c_1(x,y,t) + \sqrt{2}g e^{-a} c_2(x,y,t)$ gives the desired result. The case of \(\psi_{\lambda}^{\pm}(\bm{\mu_\lambda},t) \) and the first sum in
  $\xi_{\lambda}(\bm{\mu_{\lambda}},t)$ are dealt in the same way.
  
  For the second sum in $\xi_{\lambda}(\bm{\mu_{\lambda}},t)$, we fix \(0 \leq n < \lambda-1\) and consider the sum
   \begin{align*}
     S_{n}(t) &= \sum_{\substack{n <\beta\leq \lambda-1 \\ \beta - n \equiv 1 \pmod{2}} } \left( (e^{- t \mu_{\beta+1}} +  e^{ t \mu_{\beta+1}-2t} )-(e^{- t \mu_{\beta}} +  e^{ t \mu_{\beta}-2t}) \right) ( (e^{ t \mu_{n}} + e^{- t \mu_{n}}) - (e^{ t \mu_{n+1}} + e^{- t \mu_{n+1}}))\\
     &=  ( (e^{ t \mu_{n}} + e^{- t \mu_{n}}) - (e^{ t \mu_{n+1}} + e^{- t \mu_{n+1}})) \sum_{\substack{n<\beta\leq \lambda-1 \\ \beta - n \equiv 1 \pmod{2}} } \left( (e^{- t \mu_{\beta+1}} +  e^{ t \mu_{\beta+1}-2t} )-(e^{- t \mu_{\beta}} +  e^{ t \mu_{\beta}-2t}) \right).
   \end{align*}
   Transforming the sums into definite integrals as in the case above we see that \(S_{n}(t) \) is equal to
   \begin{gather*}
     - t^2 \left( \int_{ \mu_{n}}^{ \mu_{n+1}} (e^{-t x} + e^{t x} ) d x \right) \sum_{\substack{n<\beta\leq \lambda-1 \\ \beta - \alpha \equiv 1 \pmod{2}} }
     \left( \int_{ \mu_{\beta}}^{ \mu_{\beta+1} } e^{- t x} d x + e^{-2 t}\int_{ \mu_{\beta}}^{ \mu_{\beta+1} } e^{t x} d x  \right).
   \end{gather*}
   It follows that
   \begin{align*}
     |S_{n}(t)| &\leq |t|^2 \left(\int_{0}^{1} (e^{- a x} + e^{a x} ) d x \right) \left( \int_{0}^{1 } e^{- a x} d x + e^{-2 a}\int_{0}^{1 } e^{a x } d x \right) \\
     &\leq \left|\frac{t}{a} \right|^2 \left( e^a - e^{-a} \right) \left(1 - e^{-2 a} \right),
   \end{align*}
   with a limit interpretation for \(a = 0 \).
   It follows that
   \[
     \left\vert \xi_{\lambda}(\bm{\mu_{\lambda}},t) \right| \leq \left|\frac{2 g^2 }{1-e^{-2 t}} \right| \left|\sum_{n=0}^{\lambda-2} S_n(t) \right|
     \leq \left|\frac{2 g^2 }{1-e^{-2 t}} \right| c_3(t) \lambda,
   \]
   completing the proof.
\end{proof}

For the purpose of giving an explicit formula for the propagator it is only necessary to extend the function $K_{\text{Rabi}}  (x,y,t)$ to the line $i \R$. However, for completeness we give the holomorphic extension to a larger region in the complex plane. We remark that the region is chosen according to the principal branch of logarithm (equivalently, branch of square root).

\begin{prop} \label{prop:MeromExtK}
  For fixed $x,y \in \R$, the series defining any of the entries of the heat kernel $K_{\text{Rabi}}  (x,y,t)$ is
  uniformly convergent in compacts in the complement in the complex plane of the region
  $\bigcup_{n \in \Z} \{ t = a + i  \pi n \in \C  : \, a\leq0  \}$. In particular, $K_{\text{Rabi}}  (x,y,t)$ is (entrywise)
  holomorphic in said region.
\end{prop}

\begin{proof}
  Denote by \(\mathcal{D}\) the region given by the complement of $\bigcup_{n \in \Z} \{ t = a + i  \pi n \in \C  : \, a\leq0  \}$
  and consider a compact \(K \subset \mathcal{D}\). 
  Since \(\mathcal{D}\) does not contain any zero of \(1-e^{-2t}\), we have
  \[
    |K_0(x,y,g,t)| \leq c_0
  \]
  for some constant $c_0 \ge 0$. Similarly, we set $c_1 = \max_{t\in K}(|t|)$ and
  \[
    c_2 = \max_{t \in K}\left(\left|-2g^2\coth(\tfrac{t}2) + 4g^2 (1+e^{|t|})/(\sinh(t))\right|,\left|2g^2 \tanh(\tfrac{t}2)\right| \right),
  \]
  then we have
  \begin{align*}
    & \left|\sum_{\lambda=0}^{\infty} (t\Delta)^{\lambda} \idotsint\limits_{0\leq \mu_1 \leq \cdots \leq \mu_\lambda \leq 1} e^{-2g^2 (\coth(\tfrac{t}2))^{(-1)^\lambda}+ 4g^2 \frac{\cosh(t(1-\mu_\lambda))}{\sinh(t)}(\frac{1+(-1)^\lambda}{2}) + \xi_{\lambda}(\bm{\mu_{\lambda}},t) \pm \theta_{\lambda}(x,y,\bm{\mu_{\lambda}},t)} d \bm{\mu_{\lambda}}\right| \\
    & \quad \leq  \sum_{\lambda=0}^{\infty} ( c_1 \Delta)^{\lambda} \idotsint\limits_{0\leq \mu_1 \leq \cdots \leq \mu_\lambda \leq 1} e^{c_2 + |\xi_{\lambda}(\bm{\mu_{\lambda}},t)| +|\theta_{\lambda}(x,y,\bm{\mu_{\lambda}},t)| } d \bm{\mu_{\lambda}} \leq  e^{c_2+ c_3(x,y)} \sum_{\lambda=0}^{\infty} ( c_1 \Delta)^{\lambda}  e^{ \lambda c_4 } \idotsint\limits_{0\leq \mu_1 \leq \cdots \leq \mu_\lambda \leq 1} d \bm{\mu_{\lambda}} \\
    & \quad \leq e^{c_2+ c_3(x,y)} \sum_{\lambda=0}^{\infty} \frac{( c_1 \Delta e^{c_4})^{\lambda}  e^{ \lambda c_4 }}{k!} = e^{c_2+ c_3(x,y)+ c_1\Delta e^{c_4}},
  \end{align*}
  uniformly in $K$, where the constants $c_3$ and $c_4$ are given by Lemma \ref{lem:bound}. Therefore, we see
  that the series defining any entry of $K_{\text{Rabi}}(x,y,t)$ is bounded uniformly in $K$ by
  \[
    c_0 e^{c_2+ c_3(x,y)+ c_1\Delta e^{c_4}},
  \]
  and the result follows from Weierstrass convergence theorem since $K$ is an arbitrary compact in $\mathcal{D}$.
\end{proof}

With these preparations, we are ready to give the formula for the propagator of the QRM.

\begin{thm}
  The integral kernel $U_{\text{Rabi}}(x,y,t)$ of $e^{-i t \HRabi}$ (the propagator of QRM) is given by
  \(K_{\text{Rabi}}(x,y,i t) \). Concretely, $U_{\text{Rabi}}(x,y,t)$ is given by 
  \begin{align*}    
    &U_{\text{Rabi}}  (x,y,t) =  U_0(x,y,g,t) \Bigg[ \sum_{\lambda=0}^{\infty} (i t\Delta)^{\lambda} e^{-  2  g^2 (-i\cot(\tfrac{t}2))^{(-1)^\lambda}}
    \\
    &\qquad \times \idotsint\limits_{0\leq \mu_1 \leq \cdots \leq \mu_\lambda \leq 1}  e^{-4 i g^2 \frac{\cos(t(1-\mu_\lambda))}{\sin(t)}(\frac{1+(-1)^\lambda}{2}) + \bar{\xi}_{\lambda}(\bm{\mu_{\lambda}},t)}  
          \begin{bmatrix}
            (-1)^{\lambda} \cosh  &  (-1)^{\lambda+1} \sinh  \\
            -\sinh &  \cosh
          \end{bmatrix}
             \left( \bar{\theta}_{\lambda}(x,y,\bm{\mu_{\lambda}},t) \right) d \bm{\mu_{\lambda}} \Bigg]
  \end{align*}
  with \(\bm{\mu_0} := 0\) and \(\bm{\mu_{\lambda}}= (\mu_1,\mu_2,\cdots,\mu_\lambda)\) and \(d \bm{\mu_{\lambda}} = d \mu_1 d \mu_2 \cdots d \mu_{\lambda} \)
  for \(\lambda \geq 1\).
  Here, 
  \begin{align*}
    U_0(x,y,g,t)
    & := \frac{e^{i t(g^2+\tfrac12)}}{\sqrt{2 i \pi \sin(t)}} \exp\left( - \frac{(x^2 + y^2) \cos(t) - 2 x y}{2i \sin(t)}  \right)
  \end{align*}
  and the functions \(\bar{\theta}_{\lambda}(x,y, \bm{\mu_{\lambda}},t)\) and $\bar{\xi}_\lambda(\bm{\mu_{\lambda}},t)$ are given by
  \begin{align*} 
    \bar{\theta}_{\lambda}(x,y, \bm{\mu_{\lambda}},t) &:= \frac{2\sqrt{2} g}{i \sin(t)}\left( x \cos(t) - y \right) \left( \frac{1-(-1)^{\lambda}}{2} \right)  + i \sqrt{2}g  (x-y)  \cot(\tfrac{t}2) \\
     & \quad +   \frac{2\sqrt{2} g (-1)^{\lambda} }{i \sin(t)}  \sum_{\gamma=0}^{\lambda} (-1)^{\gamma} \Big[ x  \cos(t(1 -   \mu_{\gamma})) -  y  \cos(t \mu_{\gamma})  \Big] \nonumber \\
  \bar{\xi}_\lambda(\bm{\mu_{\lambda}},t) &:=  \frac{8g^2 }{i\sin(t)} \left(\sin(\tfrac12t(1-\mu_\lambda))\right)^2 (-1)^{\lambda}  \sum_{\gamma=0}^{\lambda} (-1)^{\gamma} \cos( t \mu_{\gamma})  \\
                &  - \frac{4 g^2  }{i\sin(t)} \sum_{\substack{0\leq\alpha<\beta\leq \lambda-1\\ \beta - \alpha \equiv 1 \pmod{2}  }}  \left( \cos(t(\mu_{\beta+1}-1)-\cos(t(\mu_{\beta}-1)) \right)  ( \cos(t  \mu_{\alpha}) - \cos(t \mu_{\alpha+1})), \nonumber 
  \end{align*}
  In the formulas above, the square root is taken according to the principal branch.
\end{thm}

We note that \(U_{\text{Rabi}}(x,y,t) \) can be written completely in terms of circular functions, in contrast with the case of \(K_{\text{Rabi}}(x,y, t) \) which is given in terms of hyperbolic functions.

\begin{thm} 
  The propagator $U_{\pm}(x,y,t, \Delta)$ of $H_\pm= \HRabi |_{\mathcal{H}_\pm}$ is given by 
  \begin{align*}
    U_{\pm}(x,y,t, \Delta)
    = & U_0(x,y,g,t)\sum_{\lambda=0}^{\infty} (i t\Delta)^{2\lambda} \bar{\Phi}^-_{2\lambda}(x,y,t) \mp U_0(x,-y,g,t) \sum_{\lambda=0}^{\infty}
        (i t\Delta)^{2\lambda+1} \bar{\Phi}^+_{2\lambda+1}(x,-y,t),
  \end{align*}
  where for $\lambda\geq1$, the function $\bar{\Phi}^\pm_{\lambda}(x,y,t)$ is given by
  \begin{align*}
    \bar{\Phi}^\pm_{\lambda}(x,y,t) := e^{-2  g^2 (-i\cot(\tfrac{t}2))^{(-1)^\lambda}} \idotsint\limits_{0\leq \mu_1 \leq \cdots \leq \mu_{\lambda} \leq 1}  e^{- 4 i g^2 \frac{\cos(t(1-\mu_\lambda))}{\sin(t)}(\frac{1+(-1)^\lambda}{2}) +  \bar{\xi}_{\lambda}(\bm{\mu_{\lambda}},t)\pm \bar{\theta}_{n}(x,y, \bm{\mu_{\lambda}},t)} d \bm{\mu_{n}}
  \end{align*}
  and  
  \begin{align*}
    \bar{\Phi}^\pm_0(x,y,t) := e^{-2 ig^2\tan\big(\frac{t}2\big) \pm\sqrt2 i g(x+y)\tan\big(\frac{t}2\big)}.
  \end{align*}
\end{thm}

\section{Spectral determinant for parity Hamiltonians and Braak's $G$-functions}
\label{sec:spectral-determinant}

The $G$-function for the QRM (and the respective ones for the parity Hamiltonians) was
originally defined by Braak \cite{B2011PRL} to establish the exact solvability of the QRM. In \cite{KRW2017}
(see also \cite{Sugi2016}) a significant relation was found between the Braak $G$-function and the spectral zeta
function of the QRM (the Mellin transform of the partition function). Actually, the $G$-function is
(up to a non-vanishing function) equal to the spectral determinant of \(\HRabi\), that is, the zeta-regularized
product associated to the spectral zeta function of the QRM. In this section, we extend this result for the case
the Hamiltonians \(H_{\pm}\) of each of the parities.

Let us start by recalling the definitions of the spectral zeta functions and the spectral determinant
specialized to the case of the QRM and the Hamiltonians of each parity.
Let
\[
  \lambda^{\pm}_1 < \lambda^{\pm}_2 \leq \lambda^{\pm}_3 \leq \ldots \leq \lambda^{\pm}_n \leq \ldots (\nearrow \infty)
\]
be the eigenvalues of \(H_{\pm}\), then the (Hurwitz-type) spectral zeta function \(\zeta^{\pm}_{\text{QRM}}(s; \tau)\) is given by the Dirichlet series
\[
  \zeta^{\pm}_{\text{QRM}}(s; \tau):= \sum_{j=1}^\infty (\lambda^{\pm}_j +\tau)^{-s}.
\]
Similarly, the spectral zeta function \(\zeta_{\text{QRM}}(s; \tau)\) of the QRM is given by
\[
  \zeta_{\text{QRM}}(s; \tau) = \zeta^{+}_{\text{QRM}}(s; \tau)+  \zeta^{-}_{\text{QRM}}(s; \tau).
\]
In both cases, it is easily verified (cf. \cite{Sugi2016}) that the zeta functions above are absolutely convergent
for \(\Re(s)>1 \) for \( \tau \in \C - \Spec({\HRabi})\). 

Fix the log-branch by $-\pi\leq \arg(\tau- \lambda_i)<\pi$. For a sequence $\mathcal{A} =\{a_i\}_{i\geq 1}, \, a_i \in \C$, by
defining the associated zeta function
\[
  \zeta_{\mathcal{A}}(s) = \sum_{n=1}^{\infty} a_n^{-s},
\]
assumed to be holomorphic at $s=0$, the zeta regularized product (cf. \cite{QHS1993TAMS}) associated to $\mathcal{A}$ is given by
\begin{equation*}
  \regprod_{i=0}^\infty a_i:= \exp\left(-\frac{d}{ds}\zeta_{\mathcal{A}}^{\pm}(s)\big|_{s=0}\right).
\end{equation*}
By introducing an auxiliary parameter, the zeta regularized product is also one of the ways to define a
function with prescribed zeros.

If the sequence $\mathcal{A}$ correspond to the eigenvalues of an operator (e.g. a Hamiltonian), the zeta
regularized product is a generalization of the characteristic polynomials of finite matrices. 
For the case of the QRM, the zeta regularized product  associated to \(\zeta_{\text{QRM}}^{\pm}(s; \tau)\) is defined by
\begin{equation*}
  \regprod_{i=0}^\infty (\tau-\lambda^{\pm}_i):= \exp\left(-\frac{d}{ds}\zeta_{\text{QRM}}^{\pm}(s; \tau)\big|_{s=0}\right),
\end{equation*}
where the product is over the eigenvalues \(\lambda_i^{\pm}\) in the spectrum of \(H_{\pm}\). Now we define the spectral determinant of the Hamiltonians \(H_{\pm} \) as
\begin{equation*}
  \det (\tau-H_{\pm}):= \regprod_{i=0}^\infty (\tau-\lambda^{\pm}_i).
\end{equation*}
The spectral determinant, as a function of $\tau$ vanishes exactly at the eigenvalues of $H_{\pm}$.

In \cite{RW2017} the authors proved that the zeta regularized product of \(\zeta_{\text{Rabi}}(s;\tau)\), equivalently the spectral determinant of  \( \HRabi\), is given (up to a non-vanishing entire function) by the complete $G$-function (called generalized $G$-function in \cite{RW2017} ) given by
\[
  \mathcal{G}(x;g,\Delta) = G_{+}(x;g,\Delta) G_{-}(x;g,\Delta) \Gamma(-x)^{-2},
\]
where \(G_{\pm}(x;g,\Delta)\) are the parity $G$-functions defined in \cite{B2011PRL,B2011PRL-OnlineSupplement} (see also Appendix \ref{sec:bargm-space-confl} and cf. \cite{LB2015JPA}).

To extend the result to the parity Hamiltonians we need some preparations. First, we need to show that the spectral
zeta function $\zeta^{\pm}_{\text{QRM}}(s; \tau)$ is holomorphic around $s = 0$. In \cite{Sugi2016}, it was shown, without using an explicit formula for the heat kernel, that $\zeta_{\text{QRM}}(s; \tau)$ extends meromorphically to the complex plane with a simple pole at \(s= 1\) (cf. \cite{IW2005a,IW2005b} for a reference to the method for the case of NCHO).

Using the Mellin transform expression of $\zeta_{\text{QRM}}(s)$ by the partition function $Z_{\text{Rabi}}(t)$ we can give another proof for the meromorphic continuation, similar to one of Riemann's original proofs for the zeta function. In addition, by the same method we obtain the analytic continuation of the parity zeta function \(\zeta_{\text{QRM}}^{\pm}(s;\tau)\). The details are given in Appendix \ref{sec:proofmero}.

\begin{thm} \label{IntRep_SZF}We have
\begin{align}\label{ContourSZF}
  \zeta_{\text{QRM}}(s;\tau)= -\frac{\Gamma(1-s)}{2\pi i}\int_\infty^{(0+)} \frac{(-w)^{s-1} \Omega(w)e^{-\tau w}}{1-e^{-w}}dw.
\end{align}
Here the contour integral is given by the path which starts at $\infty$ on the real axis, encircles the origin (with a radius smaller than $2\pi$) in the positive direction and returns to the starting point and it is assumed $|\arg(-w)|\leq \pi$. This gives a meromorphic continuation of $\zeta_{\text{QRM}}(s;\tau)$  to the whole plane where the only singularity is a simple pole with residue $2$ at $s=1$.
\end{thm}

\begin{cor}
  With the notation of Theorem \ref{IntRep_SZF}, we have
  \begin{align*}
    \zeta_{\text{QRM}}^{\pm}(s;\tau)= -\frac{\Gamma(1-s)}{4\pi i}\int_\infty^{(0+)}\left( \frac{(-w)^{s-1} \Omega(w)e^{-\tau w}}{1-e^{-w}} \mp \frac{(-w)^{s-1} \Omega_{\text{odd}}(w)e^{-\tau w}}{1+e^{-w}}  \right)dw.
  \end{align*}
  This gives a meromorphic continuation of $\zeta^{\pm}_{\text{QRM}}(s;\tau)$ to the whole plane where the only singularity is a simple pole with residue $1$ at $s=1$. \qed
\end{cor}

As an important consequence of the meromorphic continuation of $\zeta_{\text{QRM}}^{\pm}(s;\tau)$ we obtain the Weyl law for the distribution of the eigenvalues of the parity Hamiltonians \(H_{\pm}\) in the usual way (cf. \cite{IW2005a,Sugi2016}).

Let us define the spectral counting functions
\begin{align*}
  N_{\text{Rabi}}(T) &= \# \{\lambda \in \Spec(\HRabi) \, | \, \lambda \le T \}, \\
  N_{\pm}(T) &= \# \{\lambda \in \Spec(H_{\pm}) \, | \, \lambda \le T \}
\end{align*}
for \(T > 0 \).

\begin{cor}
  We have
  \[
    N_{\pm}(T) \sim \frac12 N_{\text{Rabi}}(T)  \sim T,
  \]
  as \(T \to \infty \). \qed
\end{cor}

The Weyl law shows that the positive and negative parity eigenstates are equally distributed  and supports also the original Braak conjecture concerning the number of eigenvalues (in each parity) in the consecutive intervals (see e.g. \cite{B2011PRL,KRW2017}). We note that
equality of the distribution between the parities also follows from the relation $G_-(x, g, \Delta)=G_+(x,g,-\Delta)$ between $G_\pm$-functions \cite{B2011PRL} and the properties of the constraint functions/polynomials \cite{KRW2017}.

Next, we compute the residue at the poles for the $G$-functions \(G_{\pm}(x;g,\Delta)\).

\begin{lem}\label{lem:respole}
  The residue of the $G$-function \(G_{\pm}(x;g,\Delta)\) at the (simple) pole at \(x = N  \in \Z_{\geq 0}\) is given by
  \[
    \Res_{x = N} G_{\pm}(x;g,\Delta) = \frac{\Delta^2 g^N}{2(N+1)} K_N(N;g,\Delta) G_{\pm}^{(N)}(g,\Delta). 
  \]
\end{lem}

\begin{proof}
  The result follows directly by computation and comparison with the definition of \(K_N(N;g,\Delta)\) and \(G_{\pm}^{(N)}(g,\Delta)\) (see Appendix \ref{sec:bargm-space-confl} and the proof of Proposition 6.8 of \cite{KRW2017}).
\end{proof}

Next, we show that the zeros of the complete $G$-function $\mathcal{G}_{\pm}$ for each parity defined in the following captures the complete spectrum of $H_{\pm}$.  
  \[
    \mathcal{G}_{\pm}(x;g,\Delta) := G_{\pm}(x;g,\Delta) \Gamma(-x)^{-1}. 
  \]

\begin{thm}\label{thm:eigenparity}
  There is a one-to-one correspondence between eigenvalues \(\lambda\) in \(\Spec{H_{\pm}}\) and zeros \(x = \lambda + g^2 \) of the generalized
  \(G\)-function \(\mathcal{G}_{\pm}(x;g,\Delta)\).
\end{thm}

\begin{proof}
  Let \(\lambda \in \R\) be a regular eigenvalue of \(H_{\pm}\), then by the definition \( x = \lambda + g^2 \)  is a zero of \(G_{\pm}(x;g,\Delta)\).
  Now, suppose \(\lambda= N - g^2\) is an exceptional eigenvalue of \(H_{\pm}\), then by Lemma \ref{lem:respole}, we see that at \(x = \lambda -g^2 = N \) the
  function \(G_{\pm}(x;g,\Delta)\) has a finite value, and then  \(\mathcal{G}_{\pm}(x;g,\Delta)\) vanishes by the zero of \(\Gamma(-x)^{-1}\).
  Conversely, let \(x \in \R\) be a zero of \(\mathcal{G}_{\pm}(x;g,\Delta) \). If \(x \notin \Z_{\geq0} \) then \(x \) is a zero of \(G_{\pm}(x;g,\Delta)\)  and  \(\lambda = x - g^2 \) is a regular eigenvalue of \(H_{\pm}\). If \(x = N \in \Z_{\geq 0}\), then, since the zero of \(\Gamma(-x)^{-1}\)
  at \(x=N \) is canceled by the pole of \(G_{\pm}(x;g,\Delta)\), \(x = N \) must be a zero of the residue of \(G_{\pm}(x;g,\Delta) \) at \(x=N \),
  in other words, the tuple \((g,\Delta) \) must be a zero of \(K_N(N;g,\Delta)\) or \(G_{\pm}^{(N)}(g,\Delta)\) and thus \(\lambda = N - g^2 \) is an exceptional eigenvalue
  (the Juddian or non-Juddian exceptional, respectively).
\end{proof}

  \begin{rem}
    The spectrum of the QRM can be captured by irreducible representations of $\mathfrak{sl}_2$ (cf. \cite{KRW2017,W2017JPA}).
    For instance, the Juddian (resp. non-Juddian \cite{MPS2014}) exceptional solutions are obtained from the irreducible finite dimensional (resp. lowest weight) representations. The existence of these exceptional eigenvalues inherited from the quantum harmonic oscillator (or as its ruins) which are described by the oscillator representation of $\mathfrak{sl}_2$  is the reason for the presence of the gamma factor in \(\mathcal{G}_{\pm}(x;g,\Delta)\).
\end{rem}

The meaning of Theorem \ref{thm:eigenparity} is that the complete $G$-function \(\mathcal{G}_{\pm}(x;g,\Delta)\) vanishes
exactly at the eigenvalues of \(H_{\pm}\). Immediatly it follows that it is equal, up to non-vanishing constant,
to the spectral determinant of the parity Hamiltonian.

\begin{cor} \label{cor:specdet}
  There exists a non-vanishing entire function \(c_{\pm}(\tau; g, \Delta)\) such that
  \[
    \det (\tau-H_{\pm}) = c_{\pm}(\tau; g, \Delta) \,  \mathcal{G}_{\pm}(\tau;g,\Delta). \qed
  \]
\end{cor}

We conclude by making a remark on Corollary \ref{cor:specdet}. As mentioned before, in \cite{B2011PRL}
Braak proved the integrability of the QRM by defining the $G$-function of the parity Hamiltonians \(H_{\pm}\).
Here, in Corollary \ref{cor:specdet} above, we see that the $G$-function is, up to a non-vanishing constant,
equal to the spectral determinant of \(H_{\pm}\), in other words, the zeta regularized product of the spectral zeta function \(\zeta_{\text{QRM}}^{\pm}(s;\tau)\).
The zeta regularized product of a zeta function \(\zeta(s)\) is defined when the \(\zeta(s)\) function is holomorphic in a neighborhood around \(s = 0 \)
(in case \(\zeta(s)\) has a pole at \(s = 0 \), a modified zeta regularized product may be used, cf. \cite{KW2004}).
It would be interesting to investigate the relationship between the integrability (or exact solvability) of  the Hamiltonian \(H\) of a
quantum interaction model, that is, the existence of entire solutions of the corresponding Fuchsian ODE (Bargmann model), and the existence of a zeta regularized
product for its corresponding spectral zeta function \(\zeta_{H}(s;t) \), or equivalently, the meromorphic continuation of the spectral zeta function to a region containing \(s=0\). 

\appendix

 \section{Proof of meromorphic continuation of the spectral zeta function} \label{sec:proofmero}

 In this section we provide the proof the meromorphic continuation for the spectral zeta functions
 \(\zeta_{\text{Rabi}}(s)\) and \(\zeta_{\text{QRM}}^{\pm}(s;\tau)\). In addition, we given some basic properties of
 the spectral zeta functions.
 
 We define the function $\Omega(t) = \Omega(t;\Delta,g)$ implicitly by the equation $Z_{\text{Rabi}}(t) = \frac{\Omega(t)}{1-e^{-t}}$. 
 Concretely, \(\Omega(t) \) is given by
 \[
   \Omega(t) := 2 e^{g^2 t} \Bigg[ 1 + \sum_{\lambda=1}^{\infty} (t \Delta)^{2\lambda} \idotsint\limits_{0\leq \mu_1 \leq \cdots \leq \mu_{2 \lambda} \leq 1} e^{-2g^2 \coth(\frac{t}2)+  4g^2\frac{\cosh(t(1-\mu_{2\lambda}))}{\sinh(t)} +  \xi_{2 \lambda}(\bm{\mu_{2\lambda}},t) + \psi^-_{2 \lambda}(\bm{\mu_{2 \lambda}},t)} d \bm{\mu_{2 \lambda}}  \Bigg].
 \]

 First, we prove that the function $\Omega(t)$ extends holomorphically for complex $t$.

\begin{prop}
  \label{prop:holomorphy}
  The series defining the function $\Omega(t)$ is uniformly convergent in compacts in the
  complex domain $\mathcal{D}$ consisting a union of a half plane $\Re t>0$ and a disc centered at origin with
  radius $r < \pi $. In particular, $\Omega(t)$ is a holomorphic function in the region \(\mathcal{D} \).
\end{prop}

\begin{proof}
  Let \( \mathcal{D}^* \) be the region $\mathcal{D}$ without the origin and \(\mathcal{K}\) a compact region contained in
  \(\mathcal{D^*} \).
  By using Lemma \ref{lem:bound}, we see that \(\Omega(t)\) is bounded in compacts in $\mathcal{D}*$ as in
  the proof of Proposition \ref{prop:MeromExtK}.
  
  To complete the proof, we verify the behaviour of \(\Omega(t)\) at the apparent singularity at \(t = 0 \).
  It is immediately to verify that 
  \[
   \lim_{t\to 0} -2g^2 \coth(\tfrac{t}2)+  4g^2\frac{\cosh(t(1-\mu_{2\lambda}))}{\sinh(t)} = 0, \qquad \lim_{t \to 0}  \xi_{\lambda}(\bm{\mu_{\lambda}} ,t) + \psi_{\lambda}^{-}(\bm{\mu_{\lambda}},t) = 0,
 \]
 for \(\lambda \geq 1 \) and uniformly for \(0 \leq  \mu_1 \leq \mu_2 \leq \cdots \leq \mu_\lambda \leq s \leq 1\). Thus we see that \(\Omega(0) = 2 \) and 
 the result follows from Riemann continuation theorem.
  
\end{proof}

Recall that if \(\tau > \Delta + g^2 \) we have \(\lambda_j + \tau >0 \) (resp. \(\lambda^{\pm}_j + \tau >0\) ) for any \(j \in \Z_{\geq 1} \)
 and we have the following Mellin transform representation of the spectral zeta functions. 
 \begin{align*} 
   \zeta_{\text{QRM}}(s;\tau) & = \frac1{\Gamma(s)}\int_0^\infty t^{s-1}Z_{\text{Rabi}}(t)e^{-t\tau}dt, \nonumber \\  
   \zeta_{n\text{QRM}}^{\pm}(s;\tau) & = \frac1{\Gamma(s)}\int_0^\infty t^{s-1}Z_{\rm{Rabi}}^{\pm}(t)e^{-t\tau}dt.
\end{align*}
Then, using the standard argument for the Riemann zeta function (e.g. \cite{Ivic1985,T1951}) we give the proof of the analytic
continuation of the spectral zeta function $\zeta_{\text{QRM}}(s; \tau)$.  

\begin{proof}[Proof of Theorem \ref{IntRep_SZF}]
  
  First, notice that $\lim_{t\to\infty}\Omega(t)e^{-t\tau}=0$ for $\tau > g^2+\Delta$.  (this fact is known from \cite{Sugi2016} but it may be
  proved directly by the integral expression.) Therefore we see that $\zeta_{\text{QRM}}(s;\tau)$ is analytic when $\Re(s)>1$. Now, suppose
  $\Re(s) \geq 1+\delta$ for $\delta>0$. Then it is legitimate to change the contour of the integral to get
  \[
    \int_\infty^{(0+)} \frac{(-w)^{s-1} \Omega(w)e^{-\tau w}}{1-e^{-w}}dw = \{e^{\pi(s-1)i }-e^{-\pi(s-1)i }\}\int_0^{\infty}\frac{\rho^{s-1} \Omega(\rho)e^{-\rho\tau}}{1-e^{-\rho}}d\rho.
  \]
  Hence the formula \eqref{ContourSZF} follows. Since $\Omega(w)$ is holomorphic everywhere in the path, the integral is a (single-valued) analytic
  function of $s \in \C$. The expression \eqref{ContourSZF} shows that the only possible singularities of $\zeta_{\text{QRM}}(s;\tau)$ are at the singularities of
  $\Gamma(1-s)$, i.e. at the positive integer points. Since $\zeta_{\text{QRM}}(s;\tau)$ is analytic when $\Re(s)>1$, only singularity of $\zeta_{\text{QRM}}(s;\tau)$ is at the point $s=1$. Putting $s=1$ in the integral \eqref{ContourSZF}, we obtain
  \[
    \frac1{2\pi i}\int_\infty^{(0+)} \frac{\Omega(w)e^{-\tau w}}{1-e^{-w}}dw,
  \]
  which is the residue at $w=0$ of the integrand, and this residue is $\Omega(0)=2$. It follows that
  \[
    \lim_{s\to 1} \frac{\zeta_{\text{QRM}}(s;\tau)}{\Gamma(1-s)}=-2.
  \]
  Since $\Gamma(1-s)$ has a single pole at $s=1$ with residue $-1$, we observe that the only singularity of $\zeta_{\text{QRM}}(s;\tau)$ is a simple pole with residue $2$ at $s=1$. This completes the proof of the theorem.
\end{proof}

Since Lemma \ref{lem:bound} applies for \(\psi_{\lambda}^+\), from the proof of Proposition \ref{prop:holomorphy} it is
immediate to see that \(\Omega_{\text{odd}}(t)\) defined by the equation
\[
  Z_{\rm{Rabi}}^{\pm}(t) = \frac12 \left( \frac{\Omega(t)}{1-e^{- t}} \mp \frac{\Omega_{\text{odd}}(t)}{1+e^{- t}} \right).
\]
is holomorphic in the union of the right half plane and a disk of radius one centered at the origin. Therefore,
the proof of analytic continuation extends to the spectral function for the parity Hamiltonians \(H_{\pm}\).

While not needed for the main results of this paper, we give some basic properties of the spectral zeta functions
for QRM and the Hamiltonians from the parity decomposition. In particular, we describe the values of the
spectral zeta function for QRM (and the Hamiltonians of each parity) at the negative integers, given by the so-called
Rabi-Bernoulli polynomials \cite{Sugi2016,RW2020_SV}. 

By Theorem \ref{IntRep_SZF}, it is not difficult to obtain the following identity by differentiating $n$-times with respect to $\tau$ under the integral expression \eqref{ContourSZF}. 

\begin{cor}\label{cor:DD-relation}
We have 
\[
    \frac{\partial^n}{\partial \tau^n} \zeta_{\text{QRM}}(s;\tau) = (-1)^n (s)_n \zeta_{\text{QRM}}(s+n;\tau),  
\]
where \((a)_n = a (a+1) \cdots (a+n-1) \) is the Pochhammer symbol. The same relation holds for $\zeta_{\text{QRM}}^\pm(s;\tau)$. 
\end{cor}

\begin{rem}
  For the Hurwitz zeta function \(\zeta(s;a)\), the identity $\frac{\partial^n}{\partial \tau^n} \zeta(s;\tau) = - (s)_n \zeta(s+n;\tau)$ also follows immediately from its very definition (series expression).
\end{rem}

Next, observe that in the special case $s=n \in \Z$, the quotient $\frac{(-w)^{s-1} \Omega(w)e^{-\tau w}}{1-e^{-w}}$ is a single valued function of $w$. Consequently,  by the Cauchy integral formula, we see that $\int_\infty^{(0+)} \frac{(-w)^{n-1} \Omega(w)e^{-\tau w}}{1-e^{-w}}dw$ is the residue of the integrand at $w=0$,
that is, it is the coefficient of $w^{-n}$ in $\frac{(-1)^{n-1} \Omega(w)e^{-\tau w}}{1-e^{-w}}$.

We now define the $k$th Rabi-Bernoulli polynomials $(RB)_k(\tau, g^2, \Delta)$ (according to the naming in \cite{Sugi2016}, see Remark \ref{RBP} below). Notice that when $\Delta=0$, the $k$th Rabi-Bernoulli polynomial is equal to the $g^2$-shift $B_k(\tau-g^2)$ of the $k$th Bernoulli polynomial $B_k(\tau)$.

\begin{dfn}\label{Rabi-Bernoulli}
  The $k$-th Rabi-Bernoulli polynomial $(RB)_k(\tau, g^2, \Delta) \in \R[\tau, g^2, \Delta^2]$ is defined through the equation  
  \begin{equation*}
    \frac{w\Omega(w)e^{-\tau w}}{1-e^{-w}}= 2 \sum_{k=0}^\infty\frac{(-1)^k(RB)_k(\tau, g^2, \Delta)}{k!}w^k.
  \end{equation*}
  Similarly, the $k$-th positive (resp. negative) parity Rabi-Bernoulli polynomial $(RB)_k^\pm(\tau, g^2, \Delta) \in \R[\tau, g^2, \Delta]$ is defined by the generating function
  \begin{equation*}
    \frac12\Bigg[
    \frac{w\Omega(w)e^{-\tau w}}{1-e^{-w}}
    \mp \frac{w\Omega_{\rm{odd}}(w)e^{-\tau w}}{1+e^{-w}}\Bigg]
    = \sum_{k=0}^\infty\frac{(-1)^k(RB)_k^\pm(\tau, g^2, \Delta)}{k!}w^k.
\end{equation*}
  
\end{dfn}

The special values of the spectral zeta functions at the negative integers are then obtained in the
usual way.

\begin{lem}\label{SVatNegative}
  We have, for \(k \geq 1 \),
  \begin{align*}
    \zeta_{\text{QRM}}(1-k;\tau) = -\frac2k (RB)_k(\tau, g^2, \Delta^2).
  \end{align*}
  and
  \begin{align*}
    \zeta_{\text{QRM}}^\pm(1-k;\tau)= -\frac1k (RB)_k^\pm(\tau, g^2, \Delta). 
\end{align*}
\end{lem}

\begin{proof}
  We have
  \begin{align*}
    \zeta_{\text{QRM}}(1-k;\tau) &= \frac{(-1)^{k+1}\Gamma(k)}{2\pi i}
                            \int_\infty^{(0+)} \frac1{w^{k+1}}\frac{w\Omega(w)e^{-\tau w}}{1-e^{-w}}dw \\
                          & = 2\cdot\frac{(-1)^{k+1}\Gamma(k)}{k!} (-1)^k (RB)_k(\tau, g^2, \Delta^2) \\
                          & = -\frac2k (RB)_k(\tau, g^2, \Delta^2),
  \end{align*}
  as desired. The proof for \(\zeta_{\text{QRM}}^{\pm}(1-k;\tau) \) is analogous.
\end{proof}

From the lemma above, it is obvious that
\begin{equation*}
  2 (RB)_k(\tau, g^2, \Delta^2)= (RB)_k^+(\tau, g^2, \Delta)+ (RB)_k^-(\tau, g^2, \Delta).
\end{equation*}

\begin{rem} \label{RBP}
From the expression $\zeta_{\text{QRM}}(1-k;\tau)$ above, we find that $(RB)_k(\tau, g^2, \Delta)$ is identical with the Rabi-Bernoulli polynomials $R_k(g,\Delta,\tau)$ in (1.1) of \cite{Sugi2016}:
\[
  R_k(g,\Delta,\tau)= (RB)_k(\tau, g^2, \Delta).
\]
According to the result in \cite{Sugi2016}, we have $(RB)_k(\tau, g^2, \Delta) \in \Q[g^2, \Delta^2, \tau]$. The $k$th Rabi-Bernoulli polynomial is monic and its degree with respect to the variable $\tau$ is $k$ (see also Theorem \ref{thm:Rationality_RB} below). Also, $(RB)_k(\tau, 0, 0)$ is equal to the Bernoulli polynomial $B_k(\tau)$. The coefficient $2$ appearing at the definition of the Rabi-Bernoulli polynomials is considered to be the effect of the two-by-two system Hamiltonian. 

\end{rem}

The following simple difference-differential equation satisfied by the Rabi-Bernoulli polynomials is a consequence of Lemmas \ref{SVatNegative} and \ref{cor:DD-relation}.

\begin{lem} 
  We have 
  \begin{equation}\label{DD-equation}
    \frac{\partial}{\partial \tau}(RB)_{k+1}(\tau, g^2, \Delta) = - (k+1 ) (RB)_k(\tau, g^2, \Delta)  
  \end{equation} 
  for $k=0,1,2,\ldots$  \qed
\end{lem}

The explicit formula of \(\Omega(t) \) allows us to give another proof to the rationality of the coefficients of the
Rabi-Bernoulli polynomials \((RB)_k(\tau, g^2, \Delta^2)\) (proved originally in \cite{Sugi2016}) and to extend the result to
the polynomials \((RB)_k^{\pm}(\tau, g^2, \Delta) \).

\begin{thm}\label{thm:Rationality_RB}
  The Rabi-Bernoulli polynomials \((RB)_k(\tau, g^2, \Delta^2)\) (resp. \((RB)_k^{\pm}(\tau, g^2, \Delta)\) )
  as polynomials in \(\Delta^2,g^2\) and \(\tau\) (resp. in \(\Delta,g^2\) and \(\tau\)) are rational numbers. That is, \((RB)_k(\tau, g^2, \Delta) \in \Q[g^2, \Delta^2, \tau]\). Similarly, we have 
  \( (RB)_k^{\pm}(\tau, g^2, \Delta) \in \Q[g^2, \Delta, \tau]\). Moreover, the degree of the  Rabi-Bernoulli polynomials \((RB)_k(\tau, g^2, \Delta^2)\) (resp. \((RB)_k^{\pm}(\tau, g^2, \Delta)\)) with respect to the variable $\tau$ is exactly equal to $k$.
\end{thm}

\begin{proof}
  We prove the theorem only for the polynomials \((RB)_k(\tau, g^2, \Delta^2)\) since the proof for the case of the polynomials  \((RB)_k^{\pm}(\tau, g^2, \Delta) \) is
  completely analogous.

  Let \(\lambda \geq 1 \) with \(\lambda \equiv 0 \pmod{2} \). Expanding the exponentials as power series in the variable \(t\) we see that
  \[
    e^{\phi(\mu_{2 \lambda},t)+   \xi_{2\lambda}(\bm{\mu_{2\lambda}},t) +\psi^-_{2 \lambda} (\bm{\mu_{\lambda}},t)} \in \Q[g^2,\mu_1\mu_2,\cdots,\mu_\lambda][[t]].
  \]
  Then, termwise integration yields
  \[
    (\Delta t)^{2 \lambda} \idotsint\limits_{0\leq \mu_1 \leq \cdots \leq \mu_{2 \lambda} \leq 1} e^{\phi(\mu_{2 \lambda},t)+  \xi_{2\lambda}(\bm{\mu_{2\lambda}},t) +\psi^-_{2 \lambda} (\bm{\mu_{\lambda}},t)} d \bm{\mu_{2\lambda}} \in \Q[g^2,\Delta^{2}][[t]],
  \]
  and the minimum degree of any monomial appearing in the power series is at least \(2 \lambda \). It follows that \(\Omega(t) \in \Q[g^2,\Delta^{2}][[t]] \)
  since the coefficient of any given degree \(n\) is the sum of a finite number of elements of \(\Q[g^2,\Delta^{2}]\) and similarly,
  \[
    \frac{w\Omega(w)e^{-\tau w}}{1-e^{-w}} \in \Q[g^2,\Delta^{2},\tau][[w]],
  \]
  and the result follows by comparing coefficients in the definition of \( (RB)_k(\tau, g^2, \Delta^2)\).

 We notice that the only difference on the dependence of $\Delta$ from $(RB)_k(\tau, g^2, \Delta)$ is the contribution of $\Omega_{\text{odd}}(w)$ to the definition of $(RB)_k^{\pm}(\tau, g^2, \Delta)$.  
 
Further, by Lemma \ref{cor:DD-relation}, we have 
\[
    \frac{\partial^{n}}{\partial \tau^n} \zeta_{\text{QRM}}(1-k;\tau) = (-1)^n (1-k)_n \zeta_{\text{QRM}}(1-k+n;\tau).  
\]
This shows that $\frac{\partial^{k+1}}{\partial \tau^{k+1}} \zeta_{\text{QRM}}(1-k;\tau)=0$. Hence, the degree of \((RB)_k(\tau, g^2, \Delta^2)\) with respect to $\tau$ is at most $k$.  Also, since $\zeta_{\text{QRM}}(s;\tau)$ has a simple pole at $s=1$ (with non-zero residue), and looking at the fact that radius of the circle at $1$ (for the Laurent expansion) can be taken larger than $1$ we see that $(k-n) \zeta_{\text{QRM}}(1-k+n;\tau)|_{n=k} \not=0$. It follows that $\frac{\partial^{k}}{\partial \tau^{k}} \zeta_{\text{QRM}}(1-k;\tau)\not=0$. This proves the desired result for the degree with respect to $\tau$.
\end{proof}

\begin{ex}
We give here for reader's convenience the first and second Rabi-Bernoulli polynomials which are already given in \cite{Sugi2016} (Proposition 5.2 and 6.2). Since we define the Rabi-Bernoulli polynomials by the generating function \eqref{Rabi-Bernoulli}, we can compute these polynomials (at $\tau=0$) directly from the series expansion of the partition function $Z_{\text{Rabi}}(w)=\Omega(w)/(1-e^{-w})$ (the computation is essentially equivalent with the one in \cite{Sugi2016}). Actually, note first that 
$$
wZ_{\text{Rabi}}(w)=\frac{w\Omega(w)}{1-e^{-w}}
= \Omega(w)\Big[1+ \big(\frac{w}2 - \frac{w^2}6 + \cdots\big) + \big(\frac{w}2 - \frac{w^2}6 + \cdots\big)^2+ \cdots\Big] 
$$
by taking small enough $w$. Since $\Omega(0)=2$, we have $(RB)_0(\tau, g^2, \Delta) =1$. Then, using integration (due to the relation \eqref{DD-equation}) and observing the first few terms' expansion of $\Omega(w)$ at $w=0$ gives    
\begin{align*}
(RB)_1(\tau, g^2, \Delta) & = \tau - \frac12 - g^2 \, \big(= \frac12\frac{\partial}{\partial \tau}(RB)_2(\tau, g^2, \Delta)\big),\\
(RB)_2(\tau, g^2, \Delta) & = \tau^2 -(1+2g^2)\tau+\frac16+g^2+g^4+\Delta^2.
\end{align*}
For the explicit formula for the third $(RB)_3(\tau, g^2, \Delta)$, see Proposition 6.6 in \cite{Sugi2016}. By means of this procedure, in principle, it is clear that we can compute the Rabi-Bernoulli polynomials explicitly but do not have a general formula in $k$. Thus, apart from the equation \eqref{DD-equation}, it is desirable to obtain a certain recursion formula among these Rabi-Bernoulli polynomials similarly to the Bernoulli one if any, e.g. from the Heun ODE \cite{SL2000} viewpoint. We will return this problem in the future. 
\end{ex}
  
\begin{rem}
The main idea behind the study of spectral zeta functions for quantum (interaction) models is that while detailed information about eigenvalues of a system is difficult to obtain, useful information about the complete spectrum may be elucidated from the analytic properties of the zeta function. The approach is similar to the partition function of a system as we see in this paper, and in fact, both the partition function and the spectral zeta function are intimately related via the Mellin transform. In addition, the spectral zeta function may possess interesting number theoretical properties, as in the case of the NCHO (see \cite{KW2007, LOS2016PAMS, L2018, KW2019}) by considering special values at integer points. We will study the special values of the spectral zeta function in the forthcoming paper \cite{RW2020_SV}. 
\end{rem}

 \section{Confluent Heun picture and $G$-functions of the QRM} \label{sec:bargm-space-confl}

In this Appendix we give a brief introduction to the confluent picture of the QRM via the Bargmann space and to the $G$-functions used to prove its integrability in
\cite{B2011PRL}. We follow the discussion in \cite{B2011PRL-OnlineSupplement} and suggest the reader to consult
  \cite{B2013MfI, KRW2017, Reyes2018PhD} for more details.

We introduce first the Bargmann space (or Segal-Bargmann space). We refer the reader to \cite{Sc1967AP} for an extended discussion on the application of
Bargmann space to spectral problems. In this section we use the notation \(\partial_z := \frac{d }{d z} \).

Denote by \(\mathcal{V}(\C)\) the space of entire functions \(f : \C \to \C \). In \(\mathcal{V}(\C)\) we have an inner-product defined
for \(f,g \in \mathcal{V}(\C)\) by
\begin{equation*}
  (f,g)_{\mathcal{B}} = \int_\C \overline{f(z)}g(z) d\mu(z)
\end{equation*}
where \( d \mu(z) = \frac{1}{\pi}e^{-|z|^2} dx dy\) for \(z = x + i y\), and \(d x dy  \) is the Lebesgue measure in \( \C \simeq \R^2\).

The Bargmann space \(\mathcal{B}\) is the space of functions in \(\mathcal{V}(\C)\) satisfying
\[
 \| f \|_{\mathcal{B}} =  (f,f)_{\mathcal{B}}^{1/2} = \left( \int_\C |f(z)|^2 d\mu(z) \right)^{1/2}   < \infty.
\]

It is known that the Bargmann space \(\mathcal{B}\) is a complete Hilbert space unitarily equivalent to the
\(L^2(\R)\) Hilbert space by the Stone-von Neumann theorem (the inverse of the map is the Segal-Bargmann transform).

An important property of the Bargmann space is that it contains entire functions \(f\) having asymptotic expansion of the form
\begin{equation}
  \label{eq:solasympt}
  f(z) = e^{\alpha_1 z} z^{-\alpha_0}(c_0 + c_1 z^{-1} + c_2 z^{-2} + \cdots ),  
\end{equation}
as \(z \to \infty \). In particular, normal solutions of differential equations having and unramified singular point of rank \(2\)
at infinity are included.

The creation and annihilation operators \(a\) and \(a^{\dag}\) are realized in Bargmann space respectively as the differentiation and multiplication
operators, that is
\[
  a \to \partial_z, \qquad a^{\dag} \to z.
\]

The concrete realization of \(\HRabi\) as an operator acting on \(\mathcal{H}_{\mathcal{B}} = \mathcal{B}\otimes \C^2\)  is given by
\[
  \HRabi =
  \begin{bmatrix}
    z \partial_z + \Delta & g (z + \partial_z) \\
    g (z + \partial_z) & z \partial_z- \Delta
  \end{bmatrix},
\]
from this expression it is clear that the subspaces
\begin{align*}
    \mathcal{H}_+ = \left\{
  \begin{pmatrix}
    \phi_1 \\
    \phi_2
  \end{pmatrix} \in \mathcal{H}_{\mathcal{B}} \quad
  \vert \, \phi_1 \text{ is an even function}, \phi_2 \text{ is an odd function} \right\}, \\
  \mathcal{H}_- = \left\{
  \begin{pmatrix}
    \phi_1 \\
    \phi_2
  \end{pmatrix} \in \mathcal{H}_{\mathcal{B}} \quad
  \vert \, \phi_1 \text{ is an odd function}, \phi_2 \text{ is an even function} \right\},
\end{align*}
are \(\HRabi\)-invariant subspaces of \(\mathcal{H}_{\mathcal{B}} \) and \(\mathcal{H}_+ \oplus \mathcal{H}_{-} = \mathcal{H}_{\mathcal{B}}\).

Let $(\hat{T}\psi)(z):= \psi(-z)$ $(\psi \in \mathcal{B})$ be the reflection operator acting on \( \mathcal{B}\). Then, define
the unitary operator $U$ on $\mathcal{H}_{\mathcal{B}}$ by
\[
  U:= \frac1{\sqrt2}\begin{bmatrix}1&1\\
    \hat{T}&-\hat{T}
  \end{bmatrix},
\]
and with  $C= \frac1{\sqrt2}\begin{bmatrix}1&1\\ 1&-1 \end{bmatrix}$, the Cayley transform, satisfying $C^{-1}=C^t=C$, i.e. $C^2=1$.

We obtain
\begin{align*}
(CU)^\dag \HRabi CU = \begin{bmatrix}H_+&0\\
    0&H_-
  \end{bmatrix},
\end{align*}
with
\[
  H_{\pm} = z\partial_z + g(z+\partial_z) \pm \Delta \hat{T},
\]
this is the parity decomposition of the QRM (see e.g. \cite{B2011PRL-OnlineSupplement}). 

Next, we describe the confluent Heun picture of the QRM and the $G$-function of the QRM. We refer the reader to
\cite{B2011PRL,B2011PRL-OnlineSupplement,W2015IMRN} for more details.

From our discussion above, we consider \(H_{\pm} \) as operators acting on \(\mathcal{B}\). Consider a solution of the eigenvalue problem
(time-independent Schr\"odinger equation)  for  \(H_{+}\).
Concretely, a real number \(\lambda\) is part of the spectrum of \(H_{+}\) if and only if there is a function \(\psi \in \mathcal{B}\) such that
\[
   z \partial_z \psi(z) + g( \partial_z + z) \psi(z) + \psi(-z) = \lambda \psi(z).
\]
Notice the presence of \(\psi(-z)\) due to reflection operator. Therefore, by setting \(\phi_1(z) = \psi(z) \) and \(\phi_2(z)=\psi(-z)\) and applying the change of variable \( z \to -z \) to the differential equation above we obtain the coupled system of differential equations
\begin{align} \label{eq:system1}
  (z+g) \partial_z \phi_1(z) + (g z - \lambda) + \Delta \phi_2(z) &= 0 \nonumber \\
  (z-g) \partial_z \phi_2(z) - (g z + \lambda) + \Delta \phi_1(z) &= 0.
\end{align}

This system of differential equations is equivalent to a second order confluent Heun differential equation
with two regular singularities  at \(z=g,-g \) and one unramified singularity of rank \(2\) at  \(z= \infty \), we refer the reader to \cite{SL2000} for more details on confluent Heun differential equations and singularities. As mentioned already, entire solutions of this type of differential equation have asymptotic expansion \eqref{eq:solasympt} and are thus elements of the Bargmann space. Consequently, it is left to check only the holomorphicity in the complex plane of the solutions of \eqref{eq:system1}.

Next, we consider the Frobenius solutions around the singularity \(z = g\). The exponents of the equation \eqref{eq:system1} at the
singularity are \(\sigma_1 = 0, \lambda + g^2 +1\) for \(\phi_1 \) and \(\sigma_2 = 0,\lambda+g^2 \) for \(\phi_2\). Let us consider the case \( \lambda + g^2 \not\in \Z\), here the Frobenius solutions corresponding to the exponent \(0\) lead to the expressions
\begin{align*}
  \phi_1(z) &= e^{-g z} \Delta \sum_{n=0}^{\infty} K_n(x)\left( \frac{z+g}{x-n} \right)   \\
  \phi_2(z) &= e^{- g z} \sum_{n=0}^{\infty}  K_n(x) (-z + g)^n, 
\end{align*}
where \(x = \lambda + g^2\) and \(K_n(x) \) are defined by the three term recurrence relation
\begin{equation}
  \label{eq:recKn}
  n K_n(x) = f_{n-1}(x) K_{n-1}(x) - K_{n-2},
\end{equation}
with initial condition \(K_0(x)= 1, K_1(x) = f_0(x)\) with
\[
  f_n(x) = 2 g + \frac{1}{2 g} \left( n - x + \frac{\Delta^2}{x-n} \right).
\]

The Frobenius solution \(\phi_1(z)\) (resp. \(\phi_2(-z) \)) gives an expansions of \(\psi(z)\) around
\(z = g \) ( resp. \(z=-g \)) with radius of convergence \(2 g \). The condition for the solution
\(\psi(z) \) to be entire is then
\[
  G_{+}(x;z) = \phi_2(-z) - \phi_1(z) = 0,
\]
for all \(z \in \C\). However, it is enough to check in the joint domain of \(\phi_1(z) \) and \(\phi_2(-z)\),
with holomorphicity in rest of the plane following by analytic continuation.

In particular, taking \(z= 0 \), we obtain the \(G\)-function for the Hamiltonian \(H_{+} \)
\[
  G_{+}(x) = \phi_2(0) - \phi_1(0) = \sum_{n=0}^{\infty} K_n(x) \left( 1 - \frac{\Delta}{x-n} \right)g^n,
\]
and similarly 
\[
  G_{-}(x) = \sum_{n=0}^{\infty} K_n(x) \left( 1 + \frac{\Delta}{x-n} \right)g^n.
\]

In this way, we see that solutions of the equation
\[
  G_{\pm}(x) = 0
\]
determine eigenvalues \( \lambda = x - g^2 \), with \( x \not\in \Z\). These eigenvalues constitute the {\em regular spectrum }
of the QRM and are known to be non-degenerate.

On the other hand, when the second exponent \(\sigma_1 = \lambda+g^2+1 \) of \eqref{eq:system1} at \(z=g \) is an integer \(N \in \Z_{\ge 0}\), that
is when the eigenvalue is of the form \( \lambda = N - g^2\), the Frobenius solutions corresponding to the exponent \(\sigma_i =0 \, (i=1, 2)\) 
may develop a logarithmic singularity which forces the condition
\begin{align}\label{eq:constrainteq}
  K_N(N;g,\Delta) = 0,
\end{align}
in order to obtain entire solutions. In fact, these solutions, known as {\em Juddian solutions}, have only a finite number of terms
in the power series expansion.  The condition \eqref{eq:constrainteq} (usually expression in an equivalent polynomial form, see \cite{KRW2017}) is
known as {\em constraint relation for Juddian eigenvalues} of the QRM. It is known (cf. \cite{K1985JMP}) that Juddian eigenvalues are doubly degenerate, with one solution
in each parity.

Even if the condition \eqref{eq:constrainteq} does not hold, there may be entire solutions constructed from the Frobenius solutions with respect to
the exponents \(\sigma_1 = N+1, \sigma_2 = N \). The solutions are then constructed in a manner analogous to the case of regular solutions. In this case,
the $G$-function for the {\em non-Juddian exceptional eigenvalue} \(\lambda = N-g^2 \) is given by
\[
  G^{(N)}_{\pm}(g,\Delta) = - \frac{2 (N+1)}{\Delta} + \sum_{n=N+1}^{\infty} K_n(N;g,\Delta) \left( 1 \pm \frac{\Delta}{N-m} \right) g^{n-N-1},
\]
where \(K_n(N;g,\Delta)\) satisfies \eqref{eq:recKn} with initial conditions \(K_{N+1}(N;g,\Delta)= 1 \) and \(K_{n}(N;g,\Delta)=0\) for
\(n < N \). Similar to the case of regular eigenvalues, it is known that non-Juddian exceptional eigenvalues are non-degenerate.

\section*{Acknowledgements}
This work was partially supported by JST CREST Grant Number JPMJCR14D6, Japan, and
by Grand-in-Aid for Scientific Research (C) JP16K05063 and JP20K03560 of JSPS, Japan.


\begin{flushleft}

\bigskip

 Cid Reyes-Bustos \par
 Department of Mathematical and Computing Science, School of Computing, \par
 Tokyo Institute of Technology \par
 2 Chome-12-1 Ookayama, Meguro, Tokyo 152-8552 JAPAN \par\par
 \texttt{reyes@c.titech.ac.jp}

 \bigskip

 Masato Wakayama \par
 Institute of Mathematics for Industry,\par
 Kyushu University \par
 744 Motooka, Nishi-ku, Fukuoka 819-0395 JAPAN \par
 \texttt{wakayama@imi.kyushu-u.ac.jp}

\medskip

Current address: \\
Department of Mathematics, \par
Tokyo University of Science \par
1-3 Kagurazaka, Shinjyuku-ku, Tokyo 162-8601 JAPAN \par\par
\texttt{wakayama@rs.tus.ac.jp}

\end{flushleft}

\end{document}